\documentclass[11pt]{article}
\usepackage{amsthm,amsmath,amssymb,amsfonts} 
\usepackage{epsfig} 
\usepackage{latexsym,nicefrac,bbm}
\usepackage{xspace}
\usepackage{color,fancybox,graphicx,subfigure,fullpage}
\usepackage[top=1in, bottom=1in, left=1in, right=1in]{geometry}
\usepackage{tabularx}
\usepackage{hyperref} 
\usepackage{pdfsync}
\usepackage[boxruled]{algorithm2e}
\usepackage{multicol}
\usepackage{cite,cleveref}

\newcommand{\abs}[1]{\lvert #1 \rvert}

\newcommand{\floor}[1]{\lfloor #1 \rfloor}
\DeclareMathOperator{\rank}{rank}
\DeclareMathOperator{\conv}{conv}
\DeclareMathOperator{\lcm}{lcm}

\newcommand{\cI}{\mathcal{I}}

\newcommand{\cB}{\mathcal{B}}

\newcommand{\nfrac}{\nicefrac}

\newcommand{\supp}{\mathrm{supp}}
\newcommand{\poly}{\mathrm{poly}}

\title{\bf Isolating a Vertex via Lattices:  \\ Polytopes with Totally Unimodular Faces}

\usepackage{authblk}

\author[1]{ Rohit Gurjar\thanks{
The research leading to these results has received funding from the European Community's 
Seventh Framework Programme (FP7/2007-2013) under grant agreement number 257575 and 
from the Israel Science Foundation (grant number 552/16).}}
\author[2]{Thomas Thierauf\thanks{Supported by DFG grant TH 472/4.}}
\author[3]{Nisheeth K. Vishnoi}
\affil[1]{\small California Institute of Technology, USA}
\affil[2]{\small Aalen University, Germany}
\affil[3]{\small \'{E}cole Polytechnique F\'{e}d\'{e}rale de Lausanne (EPFL), Switzerland}
\date{}

\newcommand{\B}{\mathcal{B}}

\newcommand{\C}{\mathcal{C}}
\renewcommand{\P}{\mathcal{P}}

\newcommand{\I}{\mathcal{I}}

\newcommand{\comment}[1]{}

\newcommand{\twosum}{\oplus_2}

\newcommand{\wprime}{w}

\newcommand{\set}[2]{\left\{\,#1 \mid #2\,\right\}}
\newcommand{\innerprod}[2]{\langle #1 , #2 \rangle}

\newcommand{\R}{\mathbb{R}}
\newcommand{\N}{\mathbb{N}}
\newcommand{\Z}{\mathbb{Z}}
\newcommand{\F}{\mathbb{F}}
\newcommand{\GF}{\rm{GF}}
\newcommand{\powerset}[1]{2^{#1}}

\newcommand{\norm}[1]{\left\lVert#1\right\rVert}

\newtheorem{theorem}{Theorem}[section]
\newtheorem{lemma}[theorem]{Lemma}
\newtheorem{definition}[theorem]{Definition}

\newtheorem{claim}[theorem]{Claim}
\newtheorem{observation}[theorem]{Observation}

\newtheorem{fact}[theorem]{Fact}

\newtheorem*{theorem*}{Theorem}
\newtheorem{remark}[theorem]{Remark}

\crefname{theorem}{Theorem}{Theorems}
\crefname{observation}{Observation}{Observations}
\crefname{proposition}{Proposition}{Propositions}
\crefname{claim}{Claim}{Claims}
\crefname{condition}{Condition}{Conditions}
\crefname{example}{Example}{Examples}
\crefname{fact}{Fact}{Facts}
\crefname{lemma}{Lemma}{Lemmas}
\crefname{corollary}{Corollary}{Corollaries}
\crefname{definition}{Definition}{Definitions}
\crefname{remark}{Remark}{Remarks}

\begin{document}
\maketitle

\begin{abstract}
We present a geometric approach towards derandomizing the  {Isolation Lemma} by Mulmuley, Vazirani, and Vazirani. 
In particular, our approach produces a quasi-polynomial family of weights, where each weight is an integer and  quasi-polynomially bounded, that 
can isolate a vertex in any $0/1$ polytope
for which each face lies in an affine space defined by a totally unimodular matrix.
This includes the
 polytopes given by totally unimodular constraints and generalizes the recent derandomization of the Isolation Lemma for {bipartite perfect matching}
and {matroid intersection}.
We prove our  result  by  associating a 
{lattice} to each face of the polytope and showing that if there is a totally unimodular kernel matrix for this lattice, then the number of  vectors of length within $3/2$ of the shortest vector in  it is polynomially bounded.
The proof of this latter geometric fact is combinatorial and follows from a polynomial bound on  the number of circuits of size within $3/2$ of the shortest circuit in a regular matroid. 
This is the technical core of the paper and relies on a variant of Seymour's decomposition theorem for regular matroids.
It generalizes an influential result by Karger  on the number of minimum cuts in a graph to regular matroids.

\end{abstract}

\newpage

\tableofcontents

\newpage

\section{Introduction}
The  Isolation Lemma by Mulmuley, Vazirani, and Vazirani~\cite{MVV87}
states that for any given family of subsets of a ground set $E$, if we assign random 
weights (bounded in magnitude by poly($\abs{E}$)) to the elements of $E$ then, with  high probability, the minimum weight
set in the family is unique. 
Such a weight assignment is called an \emph{isolating weight assignment}.
The lemma was introduced in the context of randomized parallel algorithms for the matching problem.
Since then it has found numerous other applications, in both algorithms and complexity: e.g., a reduction from CLIQUE to UNIQUE-CLIQUE~\cite{MVV87},
NL/poly $\subseteq \oplus$L/poly~\cite{Wig94},
NL/poly $=$ UL/poly~\cite{RA00}, an RNC-algorithm for linear matroid intersection~\cite{NSV94}, and
an RP-algorithm for disjoint paths~\cite{BH14}.
In all of these results, the Isolation Lemma is the only place where they need randomness.
Thus, if the Isolation Lemma can be derandomized, i.e., if a polynomially bounded isolating weight assignment can be 
deterministically constructed, then the aforementioned results that rely on it can also be derandomized.
In particular, it will give a deterministic parallel algorithm for matching.

A simple counting argument shows that a single weight assignment with polynomially bounded weights 
cannot be  isolating  for all possible families
of subsets of~$E$.
We can relax the question and ask if we
 can construct a poly-size list of poly-bounded weight assignments such that for each family $\B \subseteq 2^E$,
 one of the weight assignments in the list is isolating. 
 Unfortunately, even this can be shown to be impossible
 via  arguments involving the polynomial identity testing (PIT) problem. 
The PIT problem asks if an implicitly  given multivariate polynomial is identically zero. 
 Derandomization of PIT is another important consequence of derandomizing the Isolation Lemma. 
Here, the Isolation Lemma is applied to the family of monomials present in the polynomial. 
In essence, if we have a small list of weight assignments that works for all families, 
then we will have a small hitting-set for all small degree polynomials, which is impossible (see~\cite{AM08}).
Once we know that a deterministic isolation is not possible for all families,
 a natural question is to solve the isolation question for families~$\B$,
that have a succinct representation, for example, the family of perfect matchings of a graph.

For the general setting of families with succinct representations, no deterministic isolation is known, other than the trivial construction with exponentially large weights.
In fact, derandomizing the isolation lemma in this setting will imply circuit lower bounds~\cite{AM08}.
Efficient deterministic isolation is known only for very special kinds of families,
for example, perfect matchings in some special classes of graphs~\cite{DKR10, DK98, GK87, AHT07},
$s$-$t$ paths in directed graphs~\cite{BTV09,KT16,MP17}. Recently, there has been significant progress on deterministic isolation for perfect matchings in bipartite graphs~\cite{FGT16} and subsequently, in general graphs~\cite{ST17},
and matroid intersection~\cite{GT17}, which implied quasi-NC algorithms for these problems. 

Motivated by these recent works, we give a generic approach towards derandomizing the Isolation Lemma.
We show that the approach works for a large class of combinatorial polytopes
and  conjecture that it works for a significantly larger class.
For a family of sets $\B \subseteq 2^E$, define the polytope~$P(\B) \subseteq \R^E$
to be the convex hull of the indicator vectors of the sets in~$\B$.
Our main result shows that for $m:= \abs{E}$, there exists an $m^{O(\log m)}$-sized family of weight assignments on $E$ 
with weights bounded by $m^{O(\log m)}$ that is isolating for any family~$\B$ whose corresponding polytope~$P(\B)$ satisfies the following property: 
\emph{the affine space spanned by any face of~$P(\B)$ is parallel to 
 the null space of {\bf some} 
totally unimodular (TU) matrix}; see \cref{thm:TUMisolation}.
This is a black-box weight construction in the sense that it does not need the description of the family or the polytope. 

A large variety of polytopes  satisfy this property and, as a consequence, have been extensively studied in combinatorial optimization. 
The simplest such class is when the polytope $P(\cB)$ has a description $Ax \leq b$ with $A$ being a TU matrix.
Thus,  a simple consequence of our main result is a resolution to the problem of derandomizing the isolation lemma for  polytopes with TU constraints, as raised in a recent work~\cite{ST17}.
This generalizes
the isolation result
for perfect matchings in a bipartite graph~\cite{FGT16}, 
since the perfect matching polytope of a bipartite graph can be described by the incidence matrix of the graph, which is TU.
Other examples of families whose polytopes are defined by TU constraints
are 
vertex covers of a bipartite graph, independent sets of a bipartite graph, 
and, edge covers of a bipartite graph.
Note that these three problems are computationally equivalent to bipartite matching and thus, already have quasi-NC algorithms due to \cite{FGT16}.
However, the isolation results for these families are not directly implied by isolation for bipartite matchings. 

Our work also generalizes the isolation result
   for the family of common bases of two matroids~\cite{GT17}.
In the matroid intersection problem,
the constraints of the common base polytope are a rank bound on every subset of the ground set.
These constraints, in general, do not form a TUM.
However,
for every face of the polytope there exist two laminar families of subsets that form a basis for the tight constraints of the face.
The incidence matrix for the union of two laminar families is TU
(see~\cite[Theorem 41.11]{Sch03B}).

Since our condition on the polytope~$P(\B)$ does not require the constraint matrix defining the polytope itself (or any of its faces) to be TU, 
it is quite weak and is also well studied. 
Schrijver~\cite[Theorem 5.35]{Sch03A} shows that this condition is sufficient to prove that the polytope is \emph{box-totally dual integral}.
The second volume of Schrijver's book~\cite{Sch03B}
gives an excellent overview of polytopes that satisfy the condition required in \cref{thm:TUMisolation} such as 
\begin{itemize}
\item $R-S$ bibranching polytope \cite[Section 54.6]{Sch03B}
\smallskip
\item directed cut cover polytope \cite[Section 55.2]{Sch03B}
\smallskip
\item submodular flow polyhedron \cite[Theorem 60.1]{Sch03B}
\smallskip
\item lattice polyhedron \cite[Theorem 60.4]{Sch03B}
\smallskip
\item submodular base polytope \cite[Section 44.3]{Sch03B}
\smallskip
\item  many other polytopes defined via submodular and supermodular set functions
\cite[Sections 46.1,  48.1, 48.23, 46.13, 46.28, 46.29, 49.3, 49.12, 49.33, 49.39, 49.53]{Sch03B}.
\end{itemize}
We would like to point out that it is not clear if our isolation results in the above settings lead to any new derandomization of algorithms.
Finding such algorithmic applications of our isolation result would be quite interesting.

To derandomize the Isolation Lemma, we 
abstract out ideas from the bipartite matching and 
matroid intersection isolation~\cite{FGT16,GT17}, 
and give a geometric approach in terms of certain {\em lattices} associated to polytopes.
For each face $F$ of $P(\B)$, we  consider the lattice $L_F$ of all integer vectors parallel to~$F$.
We show that, if for each face $F$ of $P(\B)$, the number of near-shortest vectors in $L_F$ is polynomially bounded then 
we can construct an isolating weight assignment for $\B$ with quasi-polynomially bounded weights; see \cref{thm:isolation}.
Our main technical contribution is to give a polynomial bound on the number of near-shortest vectors in~$L_F$
(whose $\ell_1$-norm is less than $ \nfrac{3}{2} $ times the smallest $\ell_1$-norm of any vector in~$L_F$),
when this lattice is the set of integral vectors in the null space of a TUM;  see \cref{thm:tum-LA}.

The above lattice result is in contrast to general lattices where the number of such near-shortest vectors could be exponential in the dimension.

Our  result on lattices can  be  reformulated using the  language  of matroid theory:
 the number of  near-shortest circuits in a regular matroid is polynomially bounded; see \cref{thm:regular}.
In fact, we show how \cref{thm:tum-LA} can be deduced from \cref{thm:regular}.
One crucial ingredient in the proof of \cref{thm:regular} is Seymour's remarkable decomposition theorem for regular matroids~\cite{Sey80}. 
\cref{thm:regular} answers a question raised by Subramanian~\cite{Sub95} and 
is a generalization of (and builds on) known results in the case of graphic and cographic matroids,
that is, the number of near-minimum length cycles in a graph is polynomially bounded  (see~\cite{TK92,Sub95})  
and the result of Karger \cite{Kar93} that states that the number of near-mincuts in a graph is polynomially bounded.

Thus, not only do our results make  progress in derandomizing the isolation lemma for combinatorial polytopes, they make interesting connections between lattices (that are geometric objects) and combinatorial polytopes.
Our structural results about the number of near-shortest vectors in lattices 
and near-shortest circuits in matroids should be of independent interest and raise the question: 
to what extent are they generalizable?

A natural conjecture would be that for any $(0,1)$-matrix, the lattice formed by its integral null vectors has a small number of near-shortest vectors. 
In turn, this would give us the isolation result for any polytope which is defined by  a $(0,1)$-constraint matrix.
Many combinatorial polytopes have this property.
One such interesting example is the perfect matchings polytope for general graphs.
The recent result of \cite{ST17}, which showed a quasi-NC algorithm for perfect matchings, 
 does not actually go via a bound on the number of near-shortest vectors in the associated lattice. 
Obtaining a polynomial bound on this number would give a proof for their quasi-NC result in our unified framework  
and with improved parameters. 
Another possible generalization is for $(0,1)$-polytopes that have this property that the integers occurring in the description of each supporting hyperplane are bounded by a polynomial in the dimension of the polytope.
Such polytopes generalize almost all combinatorial polytopes and yet seem to have enough structure -- they have been recently studied in the context of optimization \cite{SinghV14, SV17Entropy}.

\section{Our Results}

\subsection{Isolating a vertex in a polytope}
\label{sec:TUMisolation}

For a set~$E$ and a weight function $w \colon E \to \Z$, 
we define the extension of~$w$ to any set $S \subseteq E$ by
$$w(S) := \sum_{e \in S} w(e). $$
Let $\B \subseteq 2^{E}$ be a family of subsets of $E$.
A weight function $w \colon E \to \Z$ is called \emph{isolating for}~$\B$ 
if the minimum weight set in~$\B$ 
is unique. 
In other words, the set $ \arg \min_{S \in \cB} w(S)$ is unique.
The Isolation Lemma of Mulmuley, Vazirani, and Vazirani~\cite{MVV87} asserts that a uniformly random weight function is isolating with a good probability for any~$\cB$.

\begin{lemma}[Isolation Lemma]
Let $E$ be a set, $\abs{E} = m$,
and
let $w\colon E \to \{1,2,\dots, 2m \}$ be a random weight function, 
where for each $e \in E$, 
the weight~$w(e)$ is chosen uniformly and independently at random.
Then for any family $\B \subseteq 2^{E}$, ~$w$ is isolating with probability at least~$ \nfrac{1}{2} $.
\end{lemma}

\noindent The task of derandomizing the Isolation Lemma
requires the deterministic construction of an isolating weight function with weights polynomially bounded in $m = \abs{E}$.
Here, we view the isolation question for~$\B$
as an isolation over a corresponding polytope~$P(\B)$, as follows.  	
For a set $S \subseteq E$, 
 its indicator vector
$x^S := (x^S_e)_{e \in E}$ is defined as
$$x^S_e := \begin{cases} 
		1, & \text{if } e \in S,\\
		0, & \text{otherwise.}
\end{cases}
$$
For any family of sets $\B \subseteq \powerset{E}$, 
the polytope $P(\B) \subseteq \R^m$ is defined as
 the convex hull of the indicator vectors of the sets in $\B$,
i.e., 
$$P(\B) := \conv \set{x^S}{S \in \B}.$$
Note that $P(\B)$ is contained in the $m$-dimensional unit hypercube.

The isolation question for a family $\B$ is equivalent to 
constructing a weight vector $w \in \Z^E$ such that $\langle w,  x \rangle$ has a unique minimum over $P(\B)$.
The property we need for our isolation approach is in terms of  total unimodularity of a  matrix.

\begin{definition}[Totally unimodular matrix] 
A matrix $A \in \mathbb{R}^{n \times m}$ is said to be  \emph{totally unimodular (TU)}, if every square
submatrix has determinant~$0$ or~$\pm 1$. 
\end{definition}
\noindent Our main theorem gives an efficient quasi-polynomial isolation for a family $\B$
when each face of the polytope $P(\B)$ lies in the affine space defined by a TU matrix.

\begin{theorem}[{\bf Main Result}]
\label{thm:TUMisolation}

Let $E$ be a set with $\abs{E}=m$.
Consider a class $\mathcal{C}$ of families $\B \subseteq 2^E$ that have the following property: 
 for any face~$F$ of the polytope~$P(\B)$,
there exists a TU matrix~$A_F \in \R^{n \times m}$ such that the affine space spanned by~$F$ is given by $A_Fx = b_F$
for some $b_F \in \R^n$. 
We can construct a set $W$ of $m^{O(\log m)}$ weight assignments on~$E$ with weights bounded by $m^{O(\log m)}$
such that for any family~$\B$ in the class $\mathcal{C}$, one of the weight assignments in~$W$  is isolating.
\end{theorem}

\subsection{Short vectors in lattices associated to polytopes}
\label{sec:polytopeLattice}
Our starting point towards proving \cref{thm:TUMisolation}
is a reformulation of the isolation approach for bipartite perfect matching and matroid intersection~\cite{FGT16,GT17}.
For a set~$E$ and a family $\B \subseteq 2^{E}$,
we define a lattice corresponding to each face of the polytope~$P(\B)$.
The isolation approach works when this lattice has a small number of near-shortest vectors.
For any face $F$ of $P(\B)$, consider the lattice of all integral vectors parallel to~$F$,
$$L_F := \set{ v \in \Z^E }{ v = \alpha (x_1-x_2) \text{ for some } x_1,x_2 \in F \text{ and } \alpha \in \R }.$$
The length of the shortest nonzero vector of a lattice $L$ is denoted by
$$\lambda(L) := \min \set{ \|v\| }{ 0 \neq v \in L},$$
where $\norm{\cdot}$ denotes the $\ell_1$-norm.
We prove  that if, for all faces~$F$ of~$P(\B)$ the number  of near-shortest vectors in~$L_F$
is small, then we can efficiently isolate a vertex in~$P(\B)$.

\begin{theorem}[{\bf Isolation via Lattices}]
\label{thm:isolation}
Let $E$ be a set with $|E| = m$ and let $\B \subseteq 2^{E}$ be a family such that there exists a constant $ c >1$, 
such that for any face~$F$ of polytope~$P(\B)$, 
we have
$$\left\lvert \set{v \in L_F }{ \norm{v} < c \, \lambda(L_F)} \right\rvert \leq m^{O(1)}.$$
Then one can construct a set of $m^{O(\log m)}$ weight functions with weights bounded by $m^{O(\log m)}$
such that at least one of them is isolating for~$\B$.
\end{theorem}

 \noindent
The main ingredient of the proof of \cref{thm:TUMisolation} is to show that 
the hypothesis of \cref{thm:isolation} 
is true when the lattice $L_F$ is the set of all integral vectors in the nullspace of a TU matrix.
For any $n \times m$ matrix~$A$ we define a lattice:   
$$L(A) := \set{v \in \Z^m }{ A v =0 }.$$
\begin{theorem}[{Near-shortest vectors in TU lattices}]
\label{thm:tum-LA}
For an $n \times m$ TU matrix $A$, let $\lambda := \lambda(L(A))$. Then
$$\left\lvert \set{ v \in L(A) }{ \norm{v} < \nfrac{3}{2}\,\lambda }  \right\rvert ~=~ O(m^5).$$
\end{theorem}

\noindent A similar statement can also be shown with any $\ell_p$-norm for $p\geq 2$, but with an appropriate multiplicative constant. 
\cref{thm:tum-LA} together with \cref{thm:isolation}  implies \cref{thm:TUMisolation}.

\begin{proof}[Proof of \cref{thm:TUMisolation}]
Let $F$ be a face of the polytope~$P(\B)$ and let~$A_F$ be the TU matrix associated with~$F$.
Thus $A_F x = b_F$ defines the affine span of~$F$. 
In other words, 
the set of vectors parallel to~$F$ is precisely the solution set of $A_F x =0 $
and the lattice~$L_F$ is given by~$L(A_F)$.
\cref{thm:tum-LA} implies the hypothesis of \cref{thm:isolation}
for any $L_F = L(A_F)$,
when the matrix~$A_F$ is TU. 
\end{proof}

\subsection{Near-shortest circuits in regular matroids}
The proof of  \cref{thm:tum-LA} is combinatorial and uses the language and results from matroid theory.
We refer the reader to Section~\ref{sec:matroids} for preliminaries on matroids; 
here we just recall a few basic definitions. 
A matroid is said to be \emph{represented by a matrix}~$A$, 
if its ground set is the column set of~$A$ and its independent sets
are the sets of linearly independent columns of~$A$. 
A matroid represented by a TU matrix  is said to be a \emph{regular matroid}.
A \emph{circuit} of a matroid is a minimal dependent set. 
The following is one of our main results which gives a bound on the number of near-shortest circuits in a regular matroid,
which, in turn, implies \cref{thm:tum-LA}.
Instead of the circuit size, we consider the weight of a circuit and present a more general result.

\begin{theorem}[{Near-shortest circuits in regular matroids}]
\label{thm:regular}
Let $M=(E,\I)$ be a regular matroid with $m  = \abs{E} \geq 2$ 
and let $w\colon E \to \mathbb{N}$ be a weight function.
Suppose~$M$ does not have any circuit~$C$ with  $w(C)< r$ for some number $r$.
Then 
\[
 \abs{\set{C}{C \text{ circuit in $M$ and } w(C) < \nfrac{3r}{2}}} ~\leq~ 240\, m^{5}.
\]
\end{theorem}

\begin{remark} An extension of this result would be to give a polynomial bound on
 the number of circuits of weight at most $\alpha r$
for any constant $\alpha$. Our current proof technique does not extend to this setting. 
\end{remark}

\section{Isolation via the Polytope Lattices: Proof of \cref{thm:isolation}}
\label{sec:isolation}

This section is dedicated to a proof of \cref{thm:isolation}.
That is, we give a construction of an isolating weight assignment for a family $\B \subseteq \powerset{E}$
assuming that for each face~$F$ of the corresponding polytope~$P(\B)$, 
the lattice~$L_F$ has small number of near-shortest vectors.
First, let us see how the isolation question for a family~$\B$ translates in the polytope setting.
For any weight function $w \colon E \to \Z$,
we view~$w$ as a vector in $\Z^E$ and consider the function $\langle w,x \rangle$ over the points in~$P(\B)$. 
Note that $ \langle w , x^B \rangle = w(B)$, for any $B \subseteq E$. 
Thus, 
a weight function $w \colon E \to \Z$ is isolating for a family $\B$ if and only if 
$\langle w,x \rangle$ has a unique minimum over the polytope~$P(\B)$.

Observe that for any $w \colon E \to \Z$, the points that minimize $\langle w,  x \rangle $ in $P(\B)$ 
will form a face of the polytope~$P(\B)$. 
The idea is to build the isolating weight function in rounds.
In every round,
we slightly modify the current weight function to get a smaller minimizing face. 
Our goal is to significantly reduce the dimension of the minimizing face in every round.
We stop when we reach a zero-dimensional face, i.e., we have a unique minimum weight point in~$P(\B)$. 

The following claim asserts  that if we modify the current weight function on a small scale, 
then the new minimizing face will be a subset of the current minimizing face.
In the following, we will denote the size of the set $E$ by~$m$.
\begin{claim}
\label{cla:subface}
Let $w\colon E \to \Z$ be a weight function and~$F$ be the face of~$P(\B)$ that minimizes~$w$.
Let $w'\colon E \to \{0,1,\dots, N-1\}$ be another weight function
and let~$F'$ be the face that minimizes the combined weight function $mN \, w+ w'$.
Then $F'\subseteq F$. 
\end{claim}
\begin{proof}
Consider any vertex $x \in F'$. We show that $x \in F$.
By definition of~$F'$, for any vertex $y \in P(\B)$ we have 
$$\langle mN \, w+ w' , x \rangle \leq  \langle mN \, w+ w' ,  y \rangle.$$
In other words, 
\begin{equation}
\label{eq:w'}
 \langle mN \, w+ w',   x -y \rangle \leq 0 .
\end{equation}
Since $x$ and $y$ are vertices of $\P(\B)$, we have $x, y \in \{0,1\}^m$.
Thus, $\abs{ \langle w' , x-y \rangle} < mN.$
On the other hand, if $\abs{\langle m N\, w ,  x-y \rangle }$ is nonzero then it is at least $mN$ and thus dominates $\abs{\langle w' ,  x-y \rangle}$.
Hence, for (\ref{eq:w'}) to hold, it must be that 
$$\langle m N \,  w , x-y \rangle  \leq 0.$$
It follows that $\langle w,x \rangle \leq  \langle w, y \rangle$,
and therefore $x \in F$.
\end{proof}

\noindent
Thus, in each round, we will add a new weight function to the current function using a smaller scale
and try to get a sub-face with significantly smaller dimension.
Henceforth, $N$ will be a sufficiently large number bounded by $\poly(m)$.
The following claim gives a way to go to a smaller face. 
\begin{claim}
\label{cla:parallel}
Let $F$ be the face of $P(\B)$ minimizing~$\innerprod{w}{x}$ and 
let $v \in L_F$.
Then $\innerprod{w}{v}  =0$.
\end{claim}
\begin{proof}
Since $v \in L_F$, 
we have
$v = \alpha(x_1 - x_2)$, for some $x_1,x_2 \in F$ and $\alpha \in \R$.
As 
$x_1, x_2 \in F$,
we have $\langle w,  x_1 \rangle = \langle w, x_2 \rangle$. The claim follows.
\end{proof}

\noindent
Now, let $F_0$ be the face that minimizes the current weight function $w_0$. 
Let $v$ be in $L_{F_0}$. 
Choose a new weight function $w' \in \{0,1,\dots, N-1\}^E$ such that $\innerprod{w'}{v} \neq 0.$
Let  $w_1 := m N\, w_0 + w'$ and let $F_1$ be the face that minimizes~$w_1$.  
Clearly, $\innerprod{w_1}{v} \neq 0$ and thus, by \cref{cla:parallel}, $v \not\in L_{F_1}$. 
This implies that $F_1$ is strictly contained in $F_0$.
To ensure that $F_1$ is \emph{significantly} smaller than~$F_0$,
 we choose many vectors in $L_{F_0}$, say $v_1,v_2,\dots, v_k$,
and construct a weight vector $w'$ such that for all $i \in [k]$, we have $\innerprod{w'}{v_i} \neq 0 $.
The following well-known lemma actually constructs a list of weight vectors 
such that one of them has the desired property (see \cite[Lemma 2]{FKS84}).
\begin{lemma}
\label{lem:weights}
Given $m,k,t$, let $q = mk \log t$.
In time $\poly(q)$ one can construct a 
set of weight vectors $w_1,w_2,\dots, w_q \in \{0,1,2, \dots, q\}^m$ 
such that for any set of nonzero vectors $v_1,v_2, \dots, v_k $ in $\{-(t-1), \dots, 0,1,\dots, t-1\}^m$
there exists a $j \in [q]$ such that for all $i \in [k]$ we have
$\langle w_j ,  v_i \rangle \neq 0$.
\end{lemma}
\begin{proof}
First define $w := (1,t, t^2, \dots, t^{m-1})$.
Clearly, $\langle w ,  v_i \rangle \neq 0$ for each $i$, because each coordinate of $v_i$ is less than~$t$ in absolute value.
To get a weight vector with small coordinates, we go modulo small numbers.
We consider the following weight vectors $w_j$ for $1 \leq j \leq q$:
$$ w_j := w \bmod j .$$
 We claim that this set of weight vectors has the desired property.
We know that 
$$W = \prod_{i=1}^k \langle w , v_i \rangle  \neq 0.$$
Note that the product $W$ is bounded by $t^{mk}$.
On the other hand, it is known that 
$\lcm(2,3,\dots, q) > 2^q = t^{mk}$ for all $q \geq 7$ \cite{Nai82}.
Thus, there must exist a $2 \leq j \leq q$ such that $j$ does not divide $W$.
In other words, for all $i \in [k]$
$$\langle w , v_i \rangle \not\equiv 0 \pmod{j}$$
which is the desired property.
\end{proof}

\noindent
There are two things to note about this lemma: (i) It is black-box in the sense that
we do not need to know the set of vectors $\{v_1,v_2,\dots, v_k\}$. (ii) 
We do not know a priori which function will work in the given set of functions. So, one has to
try all possibilities. 

The lemma tells us that we can ensure that $\langle w' , v \rangle \neq 0$ for polynomially many vectors~$v$ whose 
coordinates are polynomially bounded. 
Below, we formally present the weight construction.

To prove \cref{thm:isolation},
let~$c$ be the constant in the assumption of the theorem.
Let $N = m^{O(1)}$ be a sufficiently large  number and $p = \floor{\log_c (m+1)} $.
Let~$w_0\colon E \to \Z$ be a weight function such that $ \innerprod{w_0}{v} \neq 0$ for 
all nonzero $v \in \Z^E$ with $\norm{v} < c$.
For $i = 1,2, \dots, p$, define
\begin{itemize}
\item[$F_{i-1}$:] the face of $P(\B)$ minimizing $w_{i-1}$
\item[$w'_i$:] a weight vector in $\{0,1,\dots,N-1\}^E$ such that $\innerprod{w'_i}{v} \neq 0$ for 
all nonzero $v \in L_{F_{i-1}}$ with  $ \norm{v} < c^{i+1}$. 
\item[$w_{i}$:]  $m N w_{i-1} + w'_i$. 
\end{itemize}

\noindent
Observe that $F_{i} \subseteq F_{i-1}$, for each~$i$ by \cref{cla:subface}.
Hence, also for the associated lattices we have $L_{F_{i}} \subseteq L_{F_{i-1}}$.
As we show in the next claim, the choice of $w'_i$ together with \cref{cla:parallel} ensures that there are no vectors in $L_{F_i}$ with norm less than $c^{i+1}$.
\begin{claim}
\label{cla:vci}
For $i =0,1,2, \dots, p$, 
we have $\lambda(L_{F_i}) \geq c^{i+1}$.
\end{claim}

\begin{proof}
Consider a nonzero vector $v \in L_{F_i}$.
By \cref{cla:parallel},  we have
\begin{equation} 
 \innerprod{w_i}{v} = m N \innerprod{w_{i-1}}{v} + \innerprod{w'_i}{v} = 0.
\label{eq:wiv0}
\end{equation}
Since $v$ is in $L_{F_i}$, it is also in $L_{F_{i-1}}$ and 
again by \cref{cla:parallel}, we have
$ \innerprod{w_{i-1}}{v}  = 0$.
Together with~(\ref{eq:wiv0}) we conclude that
$
\innerprod{w'_i}{v} = 0.
$
By the definition of~$w'_{i}$,
this implies that 
$\norm{v} \geq c^{i+1}$.
\end{proof}

\noindent
Finally we argue that $w_p$ is isolating.
\begin{claim}
\label{cla:Fp-point}
The face $F_p$ is a point.
\end{claim}

\begin{proof}
Let $y_1,y_2 \in F_p$ be vertices and thus belong to $\{0,1\}^m$.
Then $y_1-y_2 \in L_{F_p}$ and $\norm{y_1-y_2} \leq m < c^{p+1}$.
By \cref{cla:vci}, we have that $y_1-y_2$ must be zero, i.e., $y_1 = y_2$.
\end{proof}
\noindent
The following claim, which gives bounds on the number of weight vectors we need to try
and the weights involved, finishes the proof of \cref{thm:isolation}.
\begin{claim}
\label{cla:weight-bound}
The number of possible choices for $w_p$ such that one of them is isolating for~$\B$
is $m^{O(\log m)}$. 
The weights in each such weight vector are bounded by $m^{O( \log m)}$. 
\end{claim}

\begin{proof}
To bound the weights of~$w_p$, we bound $w'_i$ for each~$i$.
By \cref{cla:vci}, 
we have $\lambda(L_{F_{i-1}}) \geq c^{i}$, for each $1\leq i \leq p$.
The hypothesis of \cref{thm:isolation} implies
$$\left\lvert \set{v \in L_{F_{i-1}}}{\norm{v} < c^{i+1}} \right\rvert \leq m^{O(1)}.$$ 
Recall that we have to ensure $ \innerprod{w'_i}{v} \neq 0$ for all nonzero vectors~$v$ in the above set.
We apply \cref{lem:weights} with $k = m^{O(1)}$.
For parameter~$t$,
note that as $\norm{v} < c^{i+1} \leq c^{p+1} \leq c(m+1)$, each coordinate of $v$ is  
less than $c(m+1)$
and therefore $t \leq c(m+1)$.
Thus,
we get  $w'_i$ with weights bounded by~$m^{O(1)}$.
Therefore the weights in $w_p$ are bounded by $m^{O(p)} = m^{O( \log m)}$.

Recall that \cref{lem:weights} actually gives a set of $m^{O(1)}$ weight vectors for possible choices of $w'_i$
and one of them has the desired property. Thus, we try all possible combinations for each $w'_i$.
This gives us a set of $m^{O(\log m)}$ possible choices for $w_p$ such that one of them is isolating for~$\B$.
\end{proof}

\section{Number of Short Vectors in Lattices: Proof of \cref{thm:tum-LA}}
\label{sec:tum-regular}
In this section, we show that \cref{thm:tum-LA} follows from \cref{thm:regular}.
We define a circuit of a matrix and show that 
to prove \cref{thm:tum-LA}, it is sufficient to upper bound the number of near-shortest circuits of a TU matrix. 
We argue that this, in turn, is implied by a bound on the number of near-shortest circuits of a regular matroid.
Just as a circuit of a matroid is a minimal dependent set, a circuit of matrix is 
a minimal linear dependency among its columns. 
Recall that for an $n \times m$ matrix~$A$,
the lattice~$L(A)$ is defined as the set of integer vectors in its kernel,
$$L(A) := \set{v \in \Z^m }{ A v =0 }.$$

\begin{definition}[Circuit] 
\label{def:tum-circuit}
For an $n\times m$ matrix~$A$, a vector $u \in L(A)$ is a \emph{circuit of}~$A$ if
\begin{itemize}
\item 
there is no nonzero $v \in L(A)$ with $\supp(v) \subsetneq \supp(u)$, and
\item 
$\gcd(u_1,u_2, \dots, u_m) =1$.
\end{itemize}
\end{definition}
\noindent
Note that if $u$ is a circuit of~$A$, then so is $-u$.
The following property of the circuits of a TU matrix is well known (see~\cite[Lemma 3.18]{Onn10}).

\begin{fact}
\label{fac:matrix-circuit}
Let $A$ be a TU matrix.
Then every circuit of~$A$ has its coordinates in $\{-1,0,1\}$.
\end{fact}
\noindent
Now, we define a notion of conformality among two vectors. 

\begin{definition}[Conformal \cite{Onn10}]
Let $u,v \in \R^m$. 
We say that~$u$ is \emph{conformal to}~$v$, denoted by
$u \sqsubseteq v $, 
if $u_iv_i \geq 0$ and $\abs{u_i} \leq \abs{v_i}$, for each $1\leq i \leq m$.
\end{definition}
\begin{observation}
\label{obs:conformal}
For vectors~$u$ and~$v$ with $u \sqsubseteq v$, we have
$\norm{v-u} = \norm{v} - \norm{u}$.
\end{observation}
\noindent
The following lemma follows from \cite[Lemma 3.19]{Onn10}.

\begin{lemma}
\label{lem:conformal-circuit}
Let $A$ be a TU matrix. 
Then for any nonzero vector $v \in L(A)$, there is 
a circuit~$u$ of~$A$ that is conformal to~$v$.
\end{lemma}
\noindent
We use the lemma to argue that any small enough vector in $L(A)$ must be a circuit.

\begin{lemma}
\label{lem:tum-small-circuit}
Let $A$ be a TU matrix and let $\lambda := \lambda(L(A))$.
Then any nonzero vector $v \in L(A)$ with $\norm{v} < 2 \lambda$ is a circuit of $A$.
\end{lemma}
\begin{proof}
Suppose $v \in L(A)$ is not a circuit of~$A$. 
We show that $\norm{v} \geq 2 \lambda$.
By \cref{lem:conformal-circuit},
there is a circuit $u$ of $A$ with $u \sqsubseteq v$.
Since $v$ is not a circuit,  $v-u \neq 0$.
Since both~$u$ and~$v-u$ are nonzero vectors in~$L(A)$, we have
$\norm{u}, \norm{v-u} \geq \lambda $.
By \cref{obs:conformal}, we have
 $\norm{v} =\norm{v-u} + \norm{u} $ and thus,
we get that $\norm{v} \geq 2 \lambda$.
\end{proof}

\noindent
Recall that a matroid represented by a TU matrix is a regular matroid (see~\cref{thm:reg}). 
The following lemma shows that
the two definitions of circuits,
1) for TU matrices and  
2) for regular matroids, 
coincide. 

\begin{lemma}
\label{lem:circuits}
Let $M=(E,\I)$ be a regular matroid, represented by a TU matrix~$A$.
Then there is a one to one correspondence between the circuits of~$M$ 
and the circuits of~$A$ (up to change of sign).
\end{lemma}

\begin{proof}
If $u \in \R^E$ is a circuit of~$A$, 
then the columns in~$A$ corresponding to the set~$\supp(u)$
are minimally dependent. Thus, the set~$\supp(u)$ is a circuit of matroid~$M$.

In the other direction, a circuit $C \subseteq E$ of matroid~$M$ is a minimal dependent set. 
Thus, the set of columns of~$A$ corresponding to~$C$ is minimally linear dependent. 
Hence, there are precisely two circuits $u , -u \in L(A)$ with their support being~$C$.
\end{proof}
\noindent
To prove \cref{thm:tum-LA}, let~$A$ be TU matrix.
By \cref{lem:tum-small-circuit},
it suffices to bound the number of near-shortest circuits of~$A$.
By \cref{lem:circuits}, the circuits of~$A$ and the circuits
of the regular matroid~$M$ represented by~$A$, coincide. 
Moreover, the size of a circuit of~$M$ is same as the $\ell_1$-norm of the corresponding circuit of~$A$,
as a circuit of $A$ has its coordinates in $\{-1,0,1\}$ by \cref{fac:matrix-circuit}.
Now \cref{thm:tum-LA} follows from \cref{thm:regular}
when we define the weight of each element being~$1$.

\section{Proof Overview of \cref{thm:regular}}
\label{sec:overview}
\cref{thm:regular} states that for a regular matroid, the number of near-shortest circuits
-- circuits whose size is at most 3/2 of the shortest circuit size -- is polynomially bounded. 
The starting point of the proof of this theorem is a remarkable result of Seymour~\cite{Sey80} which showed that every regular matroid 
can be decomposed into a set of much simpler matroids. 
Each of these building blocks for regular matroids  either
belongs to the classes of graphic and cographic matroids -- the simplest and well-known examples of regular matroids, 
or is  a special 10-element matroid $R_{10}$ (see Section~\ref{sec:matroids} for the definitions).
One important consequence of Seymour's result is a polynomial time algorithm, the only one known, for testing 
the total unimodularity of a matrix; see \cite{Sch86} (recall that a TU matrix represents a regular matroid).
Our strategy is to leverage Seymour's decomposition theorem in order to bound the number of circuits in a regular matroid. 

\subsubsection*{Seymour's Theorem and a simple inductive approach}

Seymour's decomposition involves a sequence of binary operations on matroids, each of which is either a  $1$-sum, a $2$-sum or a $3$-sum.
Formally, it states that  for every regular matroid $M$, we can build a decomposition tree -- which is a binary rooted tree --  in which the root node is the matroid $M$, 
 every node is a $k$-sum of its two children for $k=1,2$, or $3$, and  at the bottom we have 
 graphic, cographic and the $R_{10}$ matroids as the leaf nodes.
Note that the tree, in general, is not necessarily balanced and can have large depth (linear in the ground set size).

This suggests that to bound the number of near-shortest circuits in a regular matroid, 
perhaps one can use the tree structure of its decomposition, starting from the leaf nodes and arguing, inductively,
all the way up to the root. 
It is known that the number of near-shortest circuits in graphic and cographic matroids
is polynomially bounded. 
This follows from the polynomial bounds on the number of near-shortest cycles of a graph \cite{Sub95}
and on the number of near min-cuts in a graph \cite{Kar93} (\cref{thm:graphic-cographic}). 
The challenge is to show how to combine the information at an internal node.

The $k$-sum $M$ of two matroids $M_1$ and $M_2$ is defined in a way such that 
each circuit of $M$ can be built from a combination of two circuits, one from $M_1$ and another from $M_2$. 
Thus, if we have upper bounds for the number of circuits in $M_1$ and $M_2$,  
their product will give a naive upper bound for number of circuits in $M$.
Since there can be many $k$-sum operations involved, the naive product bound can quickly explode.
Hence, to keep a polynomial bound 
 we need to take a closer look at the $k$-sum operations.

\subsection*{$k$-sum operations}

\paragraph{\textbf{$1$-sum.}} A $1$-sum $M$ of two matroids $M_1$ and $M_2$ is simply their direct sum. 
That is, the ground set of $M$ is the disjoint union of the ground sets of $M_1$ and $M_2$, and  any circuit of $M$
is either a circuit of $M_1$ or a circuit of $M_2$.

The $2$-sum and $3$-sum are a bit more intricate.
It is known that the set of circuits of a matroid completely characterizes the matroid.
The $2$-sum and $3$-sum operations are defined by describing the set of circuits of the matroid obtained by the sum.
To get an intuition for the $2$-sum operation, we first describe it on two graphic matroids. 
A graphic matroid is defined with respect to a graph, where a circuit is a simple cycle in the graph.

\paragraph{\textbf{$2$-sum on graphs.}} For two graphs $G_1$ and $G_2$, their $2$-sum $G = G_1 \twosum G_2$ is any graph obtained by
identifying an edge $(u_1,v_1)$ in $G_1$ with an edge $(u_2,v_2)$ in $G_2$, that is, 
identifying $u_1$ with $u_2$ and $v_1$ with $v_2$ and then, deleting the edge $(u_1,v_1) = (u_2,v_2)$.
It would be instructive to see how a cycle in $G$, i.e., a circuit of the associated graphic matroid, looks like.
A cycle in $G$ is either a cycle in $G_1$ or in $G_2$
that avoids the edge $(u_1,v_1)=(u_2,v_2)$, or it is a union of a path $u_1 \rightsquigarrow v_1$ in $G_1$ and a 
path $v_2 \rightsquigarrow u_2$ in $G_2$.
This last possibility is equivalent to taking a symmetric difference $C_1\triangle C_2$ of two cycles $C_1$ in $G_1$ and $C_2$ in $G_2$ such that
$C_1$ passes through $(u_1,v_1)$ and $C_2$ passes through $(u_2,v_2)$.

\paragraph{\textbf{$2$-sum on matroids.}}
The $2$-sum $M_1\twosum M_2$ of two matroids $M_1$ and $M_2$ is defined analogously.
The grounds sets of $M_1$ and $M_2$, say $E_1$ and $E_2$ respectively,  
have an element in common, say $e$ (this can be achieved by identifying an element from $E_1$ with an element from $E_2$).
The sum $M_1 \twosum M_2$ is defined on the ground set $E =E_1 \Delta E_2$, the symmetric difference of the two given ground sets.
 Any circuit of the sum $M_1\twosum M_2$ is either 
a circuit in $M_1$ or in $M_2$ that avoids the common element $e$, 
or it is the symmetric difference $C_1 \triangle C_2$ of two circuits $C_1$ and $C_2$ of $M_1$ and $M_2$, respectively,
such that both $C_1$ and $C_2$ contain the common element $e$. 
 
\paragraph{\textbf{$3$-sum on matroids.}} A $3$-sum is defined similarly. 
A matroid $M$ is a $3$-sum of two matroids $M_1$ and $M_2$
if their ground sets $E_1$ and $E_2$
have a set $S$ of three elements  in
common such that $S$ is a circuit in both the matroids
and  the ground set of $M$ is the symmetric difference $E_1\triangle E_2$. 
Moreover, a circuit of $M$ is either
a circuit in $M_1$ or in $M_2$ that avoids the common elements $S$, 
or it is the symmetric difference $C_1 \triangle C_2$ of two circuits $C_1$ and $C_2$ of $M_1$ and $M_2$, respectively,
such that both $C_1$ and $C_2$ contain a common element $e$ from $S$ and no other element from $S$.

\subsection*{The inductive bound on the number of circuits}
Our proof is by a strong induction on the ground set size. 
\paragraph{Base case:} For a graphic or cographic matroid with a ground set of size $m$,
if its shortest circuit has size $r$ then the number of its circuits of size less than $\alpha r$ is at most $m^{4\alpha}$.
For the $R_{10}$ matroid, we present a constant upper bound on the number of circuits.

 \paragraph{Induction hypothesis:}
For any regular matroid with a ground set of size $m< m_0$,
 if its shortest circuit has size $r$,
then the number of its circuits of size less than $\alpha r$ is bounded by ${m}^{c\alpha}$ for some sufficiently large constant $c$.

\paragraph{Induction step:}
We prove the induction hypothesis for a regular matroid $M$ with a ground set of size $m_0$.
Let the minimum size of a circuit in $M$ be $r$. 
We want to show a bound of $m_0^{c\alpha}$ on the number of circuits in $M$ of size less than $\alpha r$.
The main strategy here is as follows: 
by Seymour's Theorem, we can write $M$ as a $k$-sum of two smaller regular matroids $M_1$ and $M_2$, with ground sets of size $m_1 <m_0$ and $m_2 <m_0$ respectively.
As the circuits of $M$ can be written as a symmetric differences of circuits of $M_1$ and $M_2$,
we derive an upper bound on the number circuits of $M$
from the corresponding bounds for $M_1$ and $M_2$, which we get from the induction hypothesis.

\vspace{1mm}
\noindent \textbf{The $1$-sum case.} 
In this case, any circuit of $M$ is either a circuit of $M_1$ or a circuit of $M_2$.
Hence, the number of circuits in $M$ of size less than $\alpha r$ is simply the sum of the number of circuits in $M_1$ and $M_2$ of size less than $\alpha r$.
Using the induction hypothesis, this sum is bounded by $m_1^{c\alpha}+ m_2^{c\alpha}$,
which is less than $m_0^{c\alpha}$ since 
 $m_0 = m_1+m_2$.
 
 \vspace{1mm}
\noindent\textbf{The $2$-sum and $3$-sum cases.} 
Let the set of common elements in the ground sets of $M_1$ and $M_2$ be $S$.
Note that $m_0 = m_1 +m_2 -\abs{S}$.
Recall from the definition of a $k$-sum that any circuit $C$ of $M$ is of the form $C_1 \triangle C_2$, where $C_1 $ and $C_2$ are circuits in $M_1$ and $M_2$ respectively, such that either \textbf{(i)} one of them, say $C_1$, has no element from $S$ and the other one $C_2$ is empty
or \textbf{(ii)} they both contain exactly one common element from $S$.
We will refer to $C_1$ and $C_2$ as projections of $C$.
Note that $\abs{C_1},\abs{C_2} \leq \abs{C}$.
In particular, if circuit $C$ is of size less than $\alpha r$, then so are its projections
$C_1$ and $C_2$. 

\vspace{1mm}
\noindent\textbf{An obstacle.}
The first step would be to bound the number of circuits $C_1$ of $M_1$ and $C_2$ of $M_2$ using the induction hypothesis. 
However, we do not have a lower bound on the minimum size of a circuit in $M_1$ or $M_2$, which is required to use the induction hypothesis. 
What we do know is that any circuit in $M_1$ or $M_2$ that does not involve elements from $S$ is also a circuit of $M$,
and thus, must have size at least $r$.
However, a circuit that involves elements from $S$ could be arbitrarily small. 
We give different solutions for this obstacle in case \textbf{(i)} and case \textbf{(ii)} mentioned above.

\vspace{1mm}
\noindent\textbf{Case (i): deleting elements in $S$.}
Let us first
consider the circuits $C_1$ of $M_1$ that do not involve elements from $S$. 
These circuits can be viewed as circuits of a new regular matroid $M_1\setminus S$ obtained by deleting the elements in $S$ from $M_1$. 
Since we know that the minimum size of a circuit in $M_1\setminus S$ is $r$, we can apply the induction hypothesis
to get a bound of $(m_1-\abs{S})^{c\alpha}$ for the number of circuits $C_1$ of $M_1\setminus S$ of size less than $\alpha r$.   
Summing this with a corresponding bound for $M_2\setminus S$ gives us a bound less than $m_0^{c \alpha}$ for the number of circuits of $M$ in case \textbf{(i)}. 

\vspace{1mm}
\noindent\textbf{Case (ii): stronger induction hypothesis.} 
 The case when circuits $C_1$ and $C_2$ contain an element from $S$ turns out to be much harder.
For this case, we actually need to strengthen our induction hypothesis.
 Let us assume that for a regular matroid of ground set size $m <m_0$, if the minimum size of a circuit that avoids a given element $\widetilde{e}$ is $r$,
then the number of circuits containing $\widetilde{e}$ and of size less than $\alpha r$ is bounded by $m^{c\alpha}$. 
This statement will also be proved by induction, but we will come to its proof later. 

Since we know that any circuit in $M_1$ (or $M_2$) that avoids elements from $S$ has size at least $r$, 
 we can use the above stronger inductive hypothesis to get a bound of $m_1^{c \alpha}$
 on the number of circuits $C_1$ in $M_1$ containing a given element from $S$ and of size less than $\alpha r$.
Similarly, we get an analogous bound of $m_2^{c \alpha}$ for  circuits $C_2$ of $M_2$.
Since $C$ can be a symmetric difference of any $C_1$ and $C_2$,  
 the product of these two bounds, that is, $(m_1m_2)^{c\alpha}$
bounds the number of circuits $C$ of $M$ of size less than $\alpha r$.
Unfortunately, this product can be much larger than $m_0^{c\alpha}$. 
Note that this product bound on the number of circuits $C$ is not really tight since $C_1$ and $C_2$ both cannot have their sizes close to $\alpha r$ simultaneously. 
This is because $C = C_1 \triangle C_2$ and thus, $\abs{C} = \abs{C_1} + \abs{C_2}-1$.
Hence, 
a better approach is to consider different cases based on the sizes of $C_1$ and $C_2$.

\vspace{2mm}
\noindent\textbf{Number of circuits $C$ when one of its projections is small.}
We first consider the case when the size of $C_1$ is very small, i.e., close to zero. 
In this case, the size of $C_2$ will be close to $\alpha r$ and we have to take the bound of $m_2^{c \alpha}$ on the number of such circuits $C_2$. 
Now, if number of circuits $C_1$ with small size is $N$ then we get a bound of $N m_2^{c \alpha}$ on the number of circuits $C$ of $M$ of this case.
Note that $N m_2^{c\alpha}$ is  dominated by $m_0^{c\alpha}$ only when $N\leq 1$, as $m_2$ can be comparable to $m$.
While $N\leq 1$ does not always hold, we show something weaker which  is true.
 
{\emph{Uniqueness of $C_1$.}}
We can show that for any element $s$ in the set of common elements $S$,  there is at most one circuit $C_1$ of size less than $r/2$
that contains $s$ and no other element from $S$.
To see this, assume that there are two such circuits $C_1$ and $C'_1$.
It is known that the symmetric difference of two circuits of a matroid is a disjoint union of some circuits of the matroid.
Thus, $C_1 \triangle C'_1$ will be a disjoint union of circuits of $M_1$.
Since $C_1 \triangle C'_1$ does not contain any element from $S$, it is also a disjoint union of circuits of $M$.
This would lead us to a contradiction because the size of $C_1 \triangle C'_1$ is less than $r$ and $M$ does not
have circuits of size less than $r$. This proves the uniqueness of $C_1$.
Our problem is still not solved since the set $S$ can have three elements in case of a $3$-sum,
and thus, there can be three possibilities for $C_1$ (i.e., N=3).

{\emph{Assigning weights to the elements.}}
To get around this problem, we use a new idea of considering matroids elements with weights.
For each element $s$ in $S$,
consider the unique circuit $C_1$ of size at most $r/2$ that contains $s$. 
In the matroid $M_2$, we assign a weight of $\abs{C_1} -1$ to the element $s$.
The elements outside $S$ get weight $1$.
The weight of element $s \in S$ signifies that if a circuit $C_2$ of $M_2$ contains $s$ then it has to be summed up with the unique
circuit $C_1$ containing $s$, which adds a weight of $\abs{C_1}-1$.
Essentially,  the circuits of the weighted matroid $M_2$ that have weight $\gamma$  will have a one-to-one correspondence with circuits $C = C_1 \triangle C_2$ of $M$ that have size $\gamma$ and have $  \abs{C_1} < r/2 $.
Hence, we can assume there are no circuits in the weighted matroid $M_2$ of weight less than $r$.
Thus, we can apply the induction hypothesis on $M_2$, but we need to further strengthen the hypothesis to a weighted version.
By this new induction hypothesis, we will get a bound of $m_2^{c \alpha }$
on the number of circuits of $M_2$ with weight less that $\alpha r$.
As mentioned above, this will bound the number of circuits $C =C_1 \triangle C_2 $ of $M$ with size less than $\alpha r$
and $\abs{C_1} < r/2$.
Note that the bound $m_2^{c \alpha }$  is smaller than the desired bound $m_0^{c\alpha}$.

\vspace{2mm}
\noindent\textbf{Number of circuits $C$ when none of its projections is small.} 
It is relatively easier to handle the other case when $C_1$ has size at least $r/2$ (and less than $\alpha r$).
In this case, $C_2$ has size less than $(\alpha- \nfrac{1}{2} )r$.
The bounds we get by the induction hypothesis for the number of circuits $C_1$ and $C_2$ 
are $m_1^{c \alpha}$ and $m_2^{c(\alpha- \nfrac{1}{2} )}$ respectively.  
Their product $m_1^{c \alpha} m_2^{c(\alpha- \nfrac{1}{2} )}$ bounds the number of circuits $C$ in this case. 
However, this product is not bounded by $m_0^{c\alpha}$.

{\emph{Stronger version of Seymour's Theorem.}}
To get a better bound we need another key idea.
Instead of Seymour's Theorem, we work with a stronger variant given by Truemper~\cite{Tru98}.
It states that any regular matroid can be written as a $k$-sum of two smaller regular matroids $M_1$ and $M_2$ for $k=1,2$ or $3$ 
such that one of them, say $M_1$, is a graphic, cographic or $R_{10}$ matroid.
The advantage of this stronger statement is that we can take a relatively smaller bound on the number of circuits of $M_1$, 
which gives us more room for the inductive argument. 
Formally, we know from above that when $M_1$  is a graphic or cographic matroid, the number of its circuits of size less than $\alpha r$ is  at most $m_1^{4 \alpha}$.
One can choose the constant $c$ in our induction hypothesis to be sufficiently large so that the product  
$m_1^{4\alpha} m_2^{c(\alpha- \nfrac{1}{2} )}$ is bounded by $m_0^{c\alpha}$.

\subsection*{A stronger induction hypothesis}
To summarize, we work with an inductive hypothesis as follows:  
If a regular matroid (with weights) has no circuits of weight less than $r$ that avoid a given set $R$ of elements
then the number of circuits of weight less than $\alpha r$ that contain the set $R$ is bounded by $m^{c\alpha}$.
As the base case, \cref{lem:graphic-set} shows this statement for the graphic and cographic case.

When we rerun the whole inductive argument with weights and with a fixed set $R$, we run into another issue.
It turns out that in the case when the size of $C_1$ is very small, our arguments above do not go through if $C_1$
has some elements from $R$.
To avoid such a situation  we use yet another  
strengthened version of Seymour's Theorem. 
It says that any regular matroid with a given element $\widetilde{e}$ can be written as a $k$-sum of two smaller regular matroids $M_1$ and $M_2$,
such that $M_1$ is a graphic, cographic or $R_{10}$ matroid and $M_2$ is a regular matroid containing $\widetilde{e}$ (\cref{thm:decomp}).
 When our $R$ is a single element set, say $\{\widetilde{e}\}$, we use this theorem to ensure that $M_1$, and thus $C_1$,
 has no elements from $R$.
 This rectifies the problem when $R$ has size $1$. 
However, as we go deeper inside the induction, the set $R$ can grow in size. 
Essentially, whenever $\alpha$ decreases by $ \nfrac{1}{2} $ in the induction, the size of $R$ grows by $1$.
Thus, we take $\alpha$ to be $ \nfrac{3}{2} $, which means that to reach $\alpha =1$ we need only one step of decrement,
and thus, the size of $R$ at most becomes $1$.
This is the reason our main theorem only deals with circuits of size less than $ \nfrac{3}{2} $ times the smallest size. 

In order to generalize this result for an arbitrary constant $\alpha$, a different method is required. This will be the subject of a follow-up work. 

\paragraph{\bf Organization of the rest of the paper.}

The remainder of the paper is dedicated to the formal proof of \cref{thm:regular}.
We first give some matroid preliminaries and Seymour's  decomposition theorem for regular matroids 
in Section~\ref{sec:matroids}.
Finally, in Section~\ref{sec:ShortCircuits}, we prove \cref{thm:regular}.

\section{Matroids}\label{sec:matroids}

In Section~\ref{sec:matroidprelim}, 
we recall some basic definitions and well-known facts about matroids
(see, for example, \cite{Oxl06,Sch03B}).
In Section~\ref{sec:seymour}, we describe Seymour's decomposition theorem for regular matroids.

\subsection{Matroids preliminaries}\label{sec:matroidprelim}

We start with some basic definitions.
\begin{definition}[Matroid]
\label{def:matroid}
A pair $M=(E,\I)$ is a \emph{matroid} if~$E$ is a finite set 
and~$\I$ is a nonempty collection of subsets of~$E$ satisfying

\begin{enumerate}
\item if $I \in \I$ and $J \subseteq I$, then $J \in \I$,
\item if $I,J \in \I$ and  $|I| < |J|$, then $I \cup \{z\} \in \I$, for some $z \in J \setminus I$.
\end{enumerate}
A subset $I$ of~$E$ is said to be {\em independent}, if~$I$ belongs to~$\cI$ and {\em dependent} otherwise.
An inclusionwise maximal independent subset of~$E$ is a {\em base} of~$M$.
An inclusionwise minimal dependent set is a \emph{circuit} of~$M$.
\end{definition}
\noindent 
We define some special classes of matroids.

\begin{definition}[Linear, binary, and regular matroid]
A matroid $M=(E,\I)$ with $m = |E|$ is \emph{linear} or \emph{representable} over some field~$\F$,
if there is a matrix $A \in \mathbb{F}^{n \times m}$, for some~$n$,
such that the collection of subsets of  the columns of $A$ that are linearly independent over $\F$ is identical to $\I$.

A matroid $M$ is \emph{binary}, if $M$ is representable over $\GF(2)$.
A matroid $M$ is \emph{regular}, if $M$ is representable over every field.
\end{definition}

It is well known that regular matroids can be characterized in terms of TU matrices.

\begin{theorem}[See~\cite{Oxl06,Sch03B}]
\label{thm:reg}
A matroid $M$ is regular if, and only if, $M$ can be represented by a TU matrix over~$\R$.
\end{theorem}

Two special classes of regular matroids are graphic matroids
and their duals, cographic matroids.

\begin{definition}[Graphic and cographic matroid]
A matroid $M=(E,\I)$ is said to be a \emph{graphic}, if there is an undirected graph $G=(V,E)$ whose edges correspond to the 
ground set~$E$ of $M$, such that $I \in \I$ if and only if~$I$ forms a forest in~$G$. 
By~$M(G)$ we denote the graphic matroid corresponding to~$G$.

The \emph{dual of}~$M$ is the matroid  $M^*=(E,\I^*)$ over the same ground set
such that a set $I \subseteq E$ is independent in~$M^*$
if and only if $E\setminus I$ contains a base set of~$M$.
A \emph{cographic matroid} is the dual of a graphic matroid.
\end{definition}

\noindent 
For $G=(V,E)$, we can represent~$M(G)$ by the vertex-edge incidence matrix $A_G \in \{0,1\}^{V \times E}$ (over $GF(2)$), 
$$A_G(v,e) = \begin{cases} 
 		1 & \text{if } e \text{ is incident on } v, \\
 		0 & \text{otherwise. } 
		\end{cases} $$
		
\begin{definition}[Graph cut and cut-set]
For a graph~$G=(V,E)$, a \emph{cut} is a partition $(V_1,V_2)$ of $V$ into two disjoint subsets. 
Any cut $(V_1,V_2)$ uniquely determines a \emph{cut-set}, the set of edges that have one endpoint in $V_1$ and the other in~$V_2$. 
The \emph{size of a cut} is the number of edges in the corresponding cut-set.
A \emph{minimum cut} is one of minimum size.
\end{definition}

\begin{fact}
\label{fac:cographic-circuits}
Let $G=(V,E)$ be  a graph.
\begin{enumerate}
\item
The circuits of the graphic matroid~$M(G)$ are exactly the simple cycles of~$G$.
\item
The circuits of the cographic matroid~$M^*(G)$ are exactly the inclusionwise minimal cut-sets of~$G$.
\end{enumerate}
\end{fact}
\noindent 
The symmetric difference of two cycles in a graph is a disjoint union of cycles. 
The analogous statement is true for binary matroids. 
\begin{fact}
\label{fac:binary}
Let $M$ be binary. 
If~$C_1$ and~$C_2$ are circuits of~$M$, 
then the symmetric difference $C_1 \triangle C_2$ is a disjoint union of circuits.
\end{fact}

\noindent
To prove \cref{thm:regular},
we have to bound the number of short circuits in regular matroids.
In \cref{lem:graphic-set},
we start by providing such a bound for graphic and cographic matroids.
The lemma is a variant of the following theorem
that bounds the number of near-shortest cycles~\cite{Sub95}
and the number of near-minimum cuts~\cite{Kar93} in a graph.

\begin{theorem}
\label{thm:graphic-cographic}
Let $G=(V,E)$ be a graph with $m \geq 1$ edges and $\alpha \geq 2$.
\begin{enumerate}
\item If $G$ has no cycles of length at most~$r$, then the number of cycles in $G$ of length at most~$\alpha r/2$ is bounded by $(2m)^{\alpha}$~{\rm\cite{Sub95}}.
\item If $G$ has no cuts of size at most~$r$, then the number of cuts in $G$ of size at most~$\alpha r/2$ is bounded by~$m^{\alpha}$~{\rm\cite{Kar93}}.
\end{enumerate}
\end{theorem}

\noindent 
We define two operations on matroids.

\begin{definition}[Deletion, contraction, minor]
Let $M=(E,\I)$ be a matroid and $e \in E$. 
The \emph{matroid obtained from $M$ by deleting $e$} is denoted by $M\setminus e $.
Its independent sets are given by the collection $\set{I \in \I}{ e \not\in I}$.

The \emph{matroid  obtained by contracting $e$} is denoted by~$M/e$.
Its independent sets are given by the collection $\set{I \subseteq E\setminus \{e\} }{ I \cup \{e\} \in \I}$.

A matroid obtained after a series of deletion and contraction operations on~$M$ is called a \emph{minor of~$M$}.
\end{definition}

\begin{fact} 
\label{fac:closed}
Let $M=(E,\I)$ be a matroid and $e \in E$. 
\begin{enumerate}
\item
The circuits of $M \setminus e$ are those circuits of $M$ that do not contain $e$. 
\item
The classes of regular matroids, graphic matroids, and cographic matroids are 
minor closed.
\end{enumerate}
\end{fact}

\noindent
For a characterization of regular matroids, we will need a specific matroid $R_{10}$, first introduced by~\cite{Bix77}.
It is a matroid, with 10 elements in the ground set, represented over $GF(2)$ by the following matrix. 
\[
\begin{pmatrix}
 1 & 1 & 0&0&1& 1&0&0&0&0 \\
1&1&1&0&0 &0&1&0&0&0 \\
0&1&1&1&0 &0&0&1&0&0 \\
0&0&1&1&1& 0&0&0&1&0 \\
1&0&0&1&1& 0&0&0&0&1 
\end{pmatrix}
\]

\begin{fact}[\cite{Sey80}]
\label{fac:R10}
Any matroid obtained by deleting some elements from $R_{10}$ is a graphic matroid.
\end{fact}

\subsection{Seymour's Theorem and its variants}
\label{sec:seymour}

The main ingredient for the proof of \cref{thm:regular} is a theorem of 
Seymour~\cite[Theorem 14.3]{Sey80} that shows that every regular matroid can be constructed from piecing together 
three kinds of matroids -- graphic matroids, cographic matroids, and the matroid $R_{10}$.
This piecing together is done via matroid operations called $1$-sum, $2$-sum and $3$-sum. 
These operations are defined for binary matroids.

\begin{definition}[Sum of two matroids \cite{Sey80}, see also \cite{Oxl06}]
\label{def:sum}
Let $M_1 = (E_1, \I_1)$ and $M_2 = (E_2, \I_2)$ be two binary matroids, and let $S = E_1 \cap E_2$. 
The \emph{sum of $M_1$ and $M_2$} is a matroid denoted by $M_1 \triangle M_2$.
It is defined over the ground set $E_1 \triangle E_2$
such that the circuits of $M_1 \triangle M_2$ are minimal non-empty subsets of $E_1 \triangle E_2$ that 
are of the form $C_1 \triangle C_2$,
where  $C_i$ is a (possibly empty) disjoint union of circuits of $M_i$, for $i=1,2$. 
\end{definition}

\noindent 
From the characterization of the circuits of a matroid~\cite[Theorem 1.1.4]{Oxl06}, 
it can be verified that the sum $M_1 \triangle M_2$ is indeed a matroid. 

We are only interested in three special sums:

\begin{definition}[$1,2,3$-sums]
Let $M_1 = (E_1, \I_1)$ and $M_2 = (E_2, \I_2)$ be two binary matroids and $E_1 \cap E_2 = S$.
Let $m_1 = \abs{E_1}$, $m_2 = \abs{E_2}$, and $s = \abs{S}$.
Let furthermore $m_1, m_2  < |E_1 \triangle E_2| = m_1 + m_2 -2s$. 
The sum $M_1 \triangle M_2$ is called  a
\begin{itemize}
\item $1$-sum, if $s=0$,
\item $2$-sum, if $s=1 $ and $S$ is not a circuit of $M_1, M_2, M^*_1$ or $M^*_2$,
\item $3$-sum, if $s=3 $ and $S$ is a circuit of $M_1$ and $M_2$ that
does not contain a circuit of $M^*_1$ or~$M^*_2$.
\end{itemize}
\end{definition}
\noindent 
Note that the condition $m_1, m_2 < m_1 + m_2 -2s$ implies that 
\begin{equation}
m_1,m_2 \geq 2s+1 
\label{eq:2s+1}
\end{equation}

\noindent
From the definition of $M_1\triangle M_2$ the following fact follows easily.
\begin{fact}
\label{cla:disjointCircuits}
Let $C_i$ be a disjoint union of circuits of~$M_i$, for $i=1,2$.
If $C_1 \triangle C_2$ is a subset of $E_1 \triangle E_2$ then it is a disjoint union of circuits of $M_1 \triangle M_2$.
\end{fact}

\noindent
In particular, it follows that for $i=1,2$, any circuit~$C_i$ of~$M_i$ with $C_i \subseteq E_i \setminus S$
is a circuit of~$M_1 \triangle M_2$.
Further, for $1$-sums, circuits are easy to characterize.
\begin{fact}[Circuits in a $1$-sum]
\label{fac:1sum-circuits}
If $M$ is a $1$-sum of~$M_1$ and~$M_2$ then any circuit of~$M$ is either a circuit of~$M_1$
or a circuit of~$M_2$. 
\end{fact}
\noindent Thus, if one is interested in the number of circuits, one can assume that the given matroid is not a $1$-sum
of two smaller matroids.

\begin{definition}[Connected matroid]
\label{def:connected}
A matroid $M$ is \emph{connected} if it cannot be written as a $1$-sum of two smaller matroids.
\end{definition}

\noindent A characterization of circuits in a 2-sum or 3-sum is not as easy.
Seymour~\cite[Lemma 2.7]{Sey80} provides a unique representation
of the circuits for these cases.

\begin{lemma}[Circuits in a $2$- or $3$-sum, \cite{Sey80}]
\label{lem:3sum-circuits}
Let $\C_1$ and $\C_2$ be the sets of circuits of 
$M_1$ and $M_2$, respectively. 
Let $M$ be a $2$- or $3$-sum of $M_1$ and $M_2$.
For $S = E_1 \cap E_2$, we have $\abs{S} = 1$ or $\abs{S} = 3$, respectively.
Then for any circuit~$C$ of~$M$, one of the following holds: 
\begin{enumerate}
\item $C \in \C_1$ and $S \cap C = \emptyset$, or 
\item $C \in \C_2$ and $S \cap C = \emptyset$, or
\item there exist unique $e \in S$, $C_1 \in \C_1$ and $C_2 \in \C_2$
 such that
$$S \cap C_1 = S \cap C_2 = \{e\}  \mbox{ and } C = C_1 \triangle C_2.$$
\end{enumerate}
\end{lemma}

\noindent 
Seymour proved the following decomposition theorem for regular matroids.
\begin{theorem}[Seymour's Theorem, \cite{Sey80}]
\label{thm:Seymour}
Every regular matroid can be obtained by means of $1$-sums, $2$-sums and $3$-sums,
starting from matroids that are graphic, cographic or $R_{10}$.
\end{theorem}

\noindent
However, to prove \cref{thm:regular}, we need a  
refined version of Seymour's Theorem that was proved  by Truemper~\cite{Tru98}.
Seymour's Theorem decomposes a regular matroid into a sum of two smaller regular matroids.
Truemper showed that one of the two smaller regular matroids can be chosen
to be graphic, cographic, or the $R_{10}$ matroid.
The theorem we write here slightly differs from the one by Truemper~\cite[Lemma 11.3.18]{Tru98}. 
A proof of \cref{thm:decomp} is presented in Appendix~\ref{sec:appendix-k-sums}.

\begin{theorem}[Truemper's decomposition for regular matroids, \cite{Tru98}]
\label{thm:decomp}
Let $M$ be a connected regular matroid, that is not graphic or cographic and is not isomorphic to $R_{10}$.
Let $\widetilde{e}$ be a fixed element of the ground set of $M$.
Then $M$ is a $2$-sum or $3$-sum of $M_1$ and $M_2$,
where $M_1$ is a graphic or cographic matroid, or a matroid isomorphic to $R_{10}$
and $M_2$  is a regular matroid that contains~$\widetilde{e}$.
\end{theorem}

\section{A Bound on the Number of near-shortest Circuits in Regular Matroids: Proof of \cref{thm:regular}}
\label{sec:ShortCircuits}

In this section, we prove our main technical tool:
in a regular matroid, 
the number of circuits that have size close to a shortest circuit is polynomially bounded (\cref{thm:regular}).
The proof
 argues along the decomposition provided by \cref{thm:decomp}.
First, we need to show
a bound on the number of circuits for 
the two base cases -- graphic and cographic matroids.

\subsection{Base Case: Graphic and cographic matroids}
\label{sec:co-graphic}

We actually prove a lemma for graphic and cographic matroids that does more -- 
it gives an
upper bound on the number of circuits that contain a fixed element of the ground set.
For a weight function $w\colon E \to \N$ on the ground set, 
the weight of any subset $C \subseteq E$ is defined as $w(C) :=\sum_{e \in C} w(e)$.

\begin{lemma}\label{lem:graphic-set}
Let $M=(E,\cI)$ be a graphic or cographic matroid, where $\abs{E} = m \geq 2$,
and $w\colon E \to \N$ be a weight function.
Let $R \subseteq E$ with $\abs{R} \leq 1$ (possibly empty) and 
$r$ be a positive integer.

If there is no circuit~$C$ in~$M$ such that $w(C)< r$ and $C \cap R = \emptyset$, 
then, for any integer $\alpha \geq 2$, the number of circuits~$C$ such that $R \subseteq C$ and $w(C) < \alpha r/2$ 
is at most $(2(m-\abs{R}))^{\alpha}$.
\end{lemma}

\begin{proof}
\textbf{Part 1: $M$ graphic}.
(See \cite{TK92,Sub95} for a similar argument as in this case.)
Let $G=(V,E)$ be the graph corresponding to the graphic matroid $M$.
By the assumption of the lemma, any cycle~$C$ in~$G$ such that $C \cap R = \emptyset$
has weight  $w(C) \geq r$.
Consider a cycle~$C$ in~$G$ with $R \subseteq C$ and $ w(C) < \alpha r/2$.
Let the edge sequence of the cycle $C$ be $(e_1,e_2,e_3, \ldots, e_{q})$ such that 
if $R$ is nonempty then $R = \{e_1\}$.
We choose $\alpha$ edges of the cycle $C$ as follows:
 Let ${i_1} = 1$ and  for $j =2,3, \dots, \alpha$, 
define $i_j$ to be the least index greater than $i_{j-1}$ (if one exists) such that 
\begin{equation}
\sum_{a=i_{j-1}+1}^{i_j} w(e_{a}) \geq r/2.
\label{eq:ijchoice}
\end{equation}
If such an index does not exists then define $i_j=q$.
Removing the edges $e_{i_1},e_{i_2},\dots,e_{i_{\alpha}}$ from $C$ gives us 
$\alpha$ paths:  for $j=1,2,\dots,\alpha-1$
$$p_j := (e_{i_{j}+1}, e_{i_{j}+2}, \dots, e_{i_{j+1}-1}),$$
and 
$$p_\alpha := (e_{i_{\alpha}+1}, e_{i_{\alpha}+2}, \dots, e_q).$$
Note that some of these paths might be empty.
By the choice of $i_j$ we know that $w(p_j) < r/2$ for $j=1,2,\dots,\alpha-1$.
Combining \eqref{eq:ijchoice} with the fact that $w(C) < \alpha r/2$, we obtain that $w(p_\alpha) < r/2$.
We associate the ordered tuple of oriented edges $(e_{i_1},e_{i_2},\dots,e_{i_\alpha})$ with the cycle $C$.

\begin{claim}
For two distinct cycles $C,C'$ in $G$, such that both contain $R$ and $w(C),w(C') < \alpha r/2$, the two associated tuples (defined as above) are different.
\end{claim}
\begin{proof}
For the sake of contradiction, assume that the associated tuples are same for both the cycles.
Thus, $C$ and $C'$ pass through $(e_{i_1},e_{i_2},\dots, e_{i_\alpha})$ with the same orientation of 
these edges. 
Further, there are  
$\alpha$ paths connecting them, say $p_1,p_2,\dots, p_\alpha$ from $C$ and $p'_1,p'_2,\dots,p'_\alpha$ from $C'$.
Since $C$ and $C'$ are distinct, for at least one $j$, it must be that $p_j\neq p'_j$.
However, since the starting points and the end points of $p_j$ and $p'_j$ are same, $p_j \cup p'_j$
contains a cycle $C''$. 
Moreover, since $w(p_j),w(p'_j) < r/2$, we can deduce that 
$w(C'') < r$. 
Finally, since neither of $p_j$ and $p'_j$ contain $e_1$, we get $C'' \cap R = \emptyset$.
This is a contradiction.
\end{proof}

\noindent
Since, each cycle $C$ with $w(C) < \alpha r/2$ and $R \subseteq C$ is associated with a different tuple, the number of such tuples 
upper bounds the number of such cycles. 
We bound the number of tuples depending on whether $R$ is empty or not. 
\begin{itemize}
\item When $R$ is empty, the number of tuples of $\alpha$ oriented edges is at most $(2m)^{\alpha}$.
\item When $R = \{e_1\}$, the number of choices for the rest of the $\alpha -1$ edges and their orientation is 
a most $(2(m-1))^{\alpha-1}$.
\end{itemize}

\noindent
\textbf{Part 2: $M$ cographic}.
Let $G=(V,E)$ be the graph corresponding to the cographic matroid~$M$ and let $n=\abs{V}$.
Recall from \cref{fac:cographic-circuits} that circuits in cographic matroids are inclusionwise minimal cut-sets in~$G$.
By the assumption of the lemma,
any cut-set~$C$ in~$G$ with $R \cap C = \emptyset$ has weight $w(C) \geq r$.
Note that this implies that~$G$ is connected, and therefore $m \geq n-1$.
We  want to give a bound on the number of cut-sets~$C \subseteq E$ 
such that $w(C) < \alpha r/2$ and $R \subseteq C$.

We argue similar to the  probabilistic construction of a minimum cut of Karger~\cite{Kar93}.
The basic idea is to contract randomly chosen edges.
\emph{Contraction of an edge} $e = (u,v)$ means 
that all edges between~$u$ and~$v$ are deleted and then~$u$ is identified with~$v$. 
Note that we get a multi-graph that way:
if there were two edges $(u,w)$ and $(v,w)$ before the contraction, 
they become two parallel edges after identifying~$u$ and~$v$.
The contracted graph is denoted by~$G/e$. 
The intuition behind contraction is,
that randomly chosen edges are likely to avoid the edges of a minimum cut.

The following algorithm implements the idea.
It does~$k \leq n$ contractions in the first phase and then chooses a random cut
within the remaining nodes of the contracted graph in the second phase that contains the edges of~$R$.
Note that any cut-set of the contracted graph is also a cut-set of the original graph.

\newcommand{\assign}{\leftarrow}

\begin{tabbing}
xxx\=xxx\=xxx\=xxx\=xxx\=xxx\= \kill
{\sc Small Cut} $(G = (V,E),R,\alpha)$ \\[0.1cm]
\emph{Contraction}\\
1 \> {\bf Repeat} $k = n-\alpha-\abs{R}$ times\\
2 \> \> {\bf randomly choose} $e \in E \setminus R$ with probability $w(e)/w(E \setminus R)$\\
3 \> \> $G \assign G/e$\\
4 \> \> $R \assign R \cup \{\text{new parallel edges to the edges in } R\}$\\[1ex]
\emph{Selection}\\
5 \> Among all possible cut-sets $C$ in the obtained graph $G$ with $R \subseteq C$, \\
 \> choose one uniformly at random and return it. 
\end{tabbing}
\noindent 
Let~$C \subseteq E$ be a cut-set with $w(C) < \alpha r/2$ and $R \subseteq C$.
We want to give a lower bound on the probability that {\sc Small Cut} outputs~$C$.

Let $G_0 = G$ and $G_i = (V_i,E_i)$ be the graph after the $i$-th contraction, 
for $i = 1,2, \dots,k$.
Note that~$G_i$ has $n_i = n-i$ nodes
since each contraction decreases the number of nodes by~$1$.
Let~$R_i$ denote the set~$R$ after the $i$-th contraction.
That is,
if $R = \{e_1\}$, 
then~$R_i$ contains all edges parallel to~$e_1$ in~$G_i$.
In case that $R = \emptyset$, also $R_i = \emptyset$.
Note that in either case $R_i \subseteq C$, if no edge of $C$ has been contracted till iteration $i$.

Conditioned on the event that no edge in~$C$ has been contracted in iterations~1 to~$i$, 
the probability that an edge from~$C$ is contracted in the $(i+1)$-th iteration is at most 
$$ w(C\setminus R_i)/w(E_i \setminus R_i).$$
We know that $ w(C\setminus R_i) \leq w(C) < \alpha r/2$. 
For a lower bound on $w(E_i \setminus R_i)$,  
consider the graph~$G'_i$ obtained from~$G_i$ by contracting the edges in~$R_i$.
The number of nodes in~$G'_i$ will be $n'_i = n -i - \abs{R}$ and 
its set of edges will be $E_i \setminus R_i$.
For any node~$v$ in~$G'_i$, consider the set~$\delta(v)$ of edges incident on~$v$ in~$G'_i$.
The set~$\delta(v)$ forms a cut-set in~$G'_i$ and also in~$G$.
Note that $\delta(v) \cap R = \emptyset$, as the edge in~$R$ has been contracted in~$G'_i$.
Thus, we can deduce that $w(\delta(v)) \geq r$. 
By summing this up for all nodes in~$G'_i$, we obtain
$$w(E_i \setminus R_i) \geq r\, n'_i/2.$$
Hence,  
$$w(E_i \setminus R_i) \geq r\, (n-i-\abs{R})/2.$$
Therefore the probability that an edge from~$C$ is contracted in the $(i+1)$-th iteration is 
$$\leq~ \frac{w(C \setminus R_i)}{w(E_i \setminus R_i)} 
~\leq~ \frac{\alpha\, r/2}{r\, (n-i-\abs{R})/2} 
~=~ \frac{\alpha}{n-i-\abs{R}}.$$
This bound becomes greater than~$1$, when $i > n-\alpha-\abs{R}$. 
This is the reason why we stop the contraction process after
$k = n-\alpha-\abs{R}$ iterations.

The probability that no edge from~$C$ is contracted in any of the rounds is
\begin{eqnarray*}
&\geq& \prod_{i=0}^{k-1} \left( 1-\frac{\alpha}{n-i-\abs{R}} \right)\\
&=& \prod_{i=0}^{k-1} \left( 1-\frac{\alpha}{k + \alpha-i} \right)\\
&=& \prod_{i=0}^{k-1} \frac{k-i}{k + \alpha-i}\\
&=& \frac{1}{{{k+\alpha} \choose k}}
 \\
&=& \frac{1}{{{n-\abs{R}} \choose \alpha}}.
\end{eqnarray*}
After $n-\alpha-\abs{R}$ contractions we are left with $\alpha+\abs{R}$ nodes. 
We claim that the number of possible cut-sets on these nodes that contain~$R$ is~$2^{\alpha-1}$.
%
In case when $R = \emptyset$, then the number of partitions of $\alpha$ nodes into two sets is clearly~$2^{\alpha-1}$.
When $R = \{e_1\}$, then the number of partitions of $\alpha+1$ nodes, 
such that the endpoints of~$e_1$ are in different parts, is again~$2^{\alpha-1}$.
 We choose one of these cuts randomly. 
Thus, the probability that $C$ survives the \emph{contraction} process and is also chosen in the 
\emph{selection} phase is 
at least 
$$ \frac{1}{2^{\alpha-1} {{n-\abs{R}} \choose \alpha} } \geq \frac{1}{({n-\abs{R}})^\alpha}.$$
Note that in the end we get exactly one cut-set. 
Thus, the number of cut-sets $C$ of weight $< \alpha r/2$ and $R \subseteq C$
must be at most $(n-\abs{R})^{\alpha}$, which is bounded by $(2(m-\abs{R}))^\alpha$
because $m \geq n-1$.
\end{proof}

\subsection{General regular matroids}
\label{sec:proof-3rby2}
In this section, we prove our main result about regular matroids.

\begin{theorem*}[\cref{thm:regular}]
Let $M=(E,\I)$ be a regular matroid with $m  = \abs{E} \geq 2$ 
and $w\colon E \to \mathbb{N}$ be a weight function.
Suppose~$M$ does not have any circuit~$C$ such that  $w(C)< r$, for some number~$r$.
Then 
\[
 \abs{\set{C}{C \text{ circuit in $M$ and } w(C) < 3r/2}} ~\leq~ 240\, m^{5}.
\]
\end{theorem*}

\begin{proof}
The proof
is by an induction on~$m$,
the size of the ground set.
For the base case, let $m \leq 10$. 
There are at most $2^{m}$ circuits in $M$.
This number is bounded by $240\, m^{5}$, for any $2 \leq m \leq 10$.

For the inductive step,
let~$M = (E,{\mathcal I})$ be a regular matroid with $\abs{E} = m > 10$
and assume that the theorem holds for all smaller regular matroids.
Note that~$M$ cannot be~$R_{10}$ since $m>10$.
We can also assume that matroid~$M$ is neither graphic nor cographic,
otherwise the bound follows from \cref{lem:graphic-set}.
By \cref{thm:Seymour}, 
matroid~$M$ can be written as a 1-, 2-, or 3-sum of 
two regular matroids~$M_1 = (E_1,{\mathcal I}_1)$ and~$M_2 = (E_2,{\mathcal I}_2)$.
We define
\begin{eqnarray*}
S &:=& E_1 \cap E_2,\\
s &:=& \abs{S},\\
m_i &:=& \abs{E_i}, \text{ for } i =1,2,\\
\C_i &:=& \set{C}{C \text{ is a circuit of } M_i}.
\end{eqnarray*}
In case that $M$ is the 1-sum of~$M_1$ and~$M_2$,
we have $S = \emptyset$, and therefore $m = m_1 + m_2$.
By \cref{fac:1sum-circuits}, 
the set of circuits of~$M$ is the union of the sets of circuits of~$M_1$ and~$M_2$.
From the induction hypothesis, we have that $M_i$ has at most $240\, m_i^5$ circuits of weight less than $3r/2$, for $i=1,2$.
For the number of such circuits in~$M$ we get the bound of
\[
240\, m_1^5 + 240\, m_2^5 \leq 240\, m^5 .
\]
This proves the theorem in case of a 1-sum.
Hence, in the following it remains to consider the case that~$M$ cannot be written as a 1-sum.
In other words,
we may assume that~$M$ is connected (\cref{def:connected}).

Now we can apply \cref{thm:decomp} and assume 
that~$M$ is a $2$- or $3$-sum of~$M_1$ and~$M_2$,
where~$M_1$ is a graphic, cographic or the~$R_{10}$ matroid,
and~$M_2$ is a regular matroid.

We define for $i=1,2$ and $e \in S$
\begin{eqnarray*}
\C_{i,e} &:=& \set{C}{C \in \C_i \text{ and } C \cap S = \{e\}},\\
M'_i &:=& M_i \setminus S,\\
\C'_i &:=& \set{C}{C \text{ is a circuit of } M'_i}.
\end{eqnarray*}
By \cref{fac:closed,fac:R10},
matroid~$M'_1$ is graphic or cographic, and $M'_2$ is regular.
Recall from \cref{lem:3sum-circuits} that 
any circuit~$C$ of~$M$ can be uniquely written as $C_1 \triangle C_2$ such that one of 
the following holds: 
\begin{itemize}
\item $C_1 =\emptyset$ and $C_2 \in \C'_2$.
\item $C_2 =\emptyset$ and $C_1 \in \C'_1$.
\item $C_1 \in \C_{1,e}$, and  $C_2 \in \C_{2,e}$,  for some  $e \in S$.
\end{itemize}
Thus, we will view each circuit~$C$ of~$M$ as $C_1 \triangle C_2$
and consider cases based on 
how the weight of~$C$ is distributed among~$C_1$ and~$C_2$.
Recall that the weight function~$w$ is defined on $E = E_1 \triangle E_2$.
We extend~$w$ to a function on $E_1 \cup E_2$ by defining
\[
w(e) = 0, \text{ for } e \in S.
\]
Now, for the desired upper bound, 
we will divide the set of circuits of~$M$ of weight less than~$3r/2$ into three cases.

\begin{description}
\item[{\bf Case 1.}] $C_1 \in \C'_{1}$.
\item[{\bf Case 2.}] $\wprime(C_1) < r/2$. This includes the case that $C_1 = \emptyset$.
\item[{\bf Case 3.}] $\wprime(C_1) \geq r/2$ and $C_2 \neq \emptyset$.
\end{description}

In the following,
we will derive an upper bound for the number of circuits in each of the three cases.
Then the sum of these bounds will be an upper bound on the number of circuits in~$M$.
We will show that the sum is less than $240\, m^5$.

\subsection*{Case 1: $C_1 \in \C'_{1}$} 
We have $C_2 = \emptyset$ and $C = C_1 \in \C'_1$. 
That is, we need to bound the number of circuits of $M'_1$.
Recall that any circuit of $M'_1$ is also a circuit of $M$.
Hence, we know there is no circuit $C_1$ in $M'_1$ with $w(C_1) < r$.
Since $M'_1$ is graphic or cographic, from \cref{lem:graphic-set}, 
the number of circuits $C_1$ of $M'_1$ with $\wprime(C_1) < 3r/2$
is at most 
$(2(m_1-s))^{3}.$
Recall from (\ref{eq:2s+1}) that $m_1 \geq 2s+1$. 
For any $m_1 \geq 2s+2$, one can verify that 
$$(2(m_1-s))^{3} \leq  240\,(m_1-2s)^{5} =: T_0.$$
On the other hand, when $m_1 = 2s+1$, 
the number of circuits can be at most $2^{m_1-s} \leq 2^4$, 
which is again bounded by~$ T_0 $.

\subsection*{Case 2: $\wprime(C_1) < r/2$}

The main point why we distinguish case~2 is that here~$C_1$ is uniquely determined.

\begin{claim}
\label{cla:unique}
For any $e \in S$,
there is at most one circuit $C_1 \in \C_{1,e}$ with 
$\wprime(C_1) < r/2$.
\end{claim}
\begin{proof}
For the sake of contradiction, assume that there are two circuits~$C_1, C'_1 \in \C_{1,e}$,
with $\wprime(C_1), \wprime(C'_1) < r/2$.
By \cref{fac:binary}, 
we know that $C_1 \triangle C'_1$ is a disjoint union of circuits in~$M_1$. 
Note that $C_1 \cap S = C'_1 \cap S = \{e\}$,
and hence $(C_1 \triangle C'_1) \cap S = \emptyset$.
Thus, $C_1 \triangle C'_1$ is in fact a disjoint union of circuits in~$M$.
Let $\widetilde{C}$ be a subset of $C_1 \triangle C'_1$ that is a circuit.
For the weight of~$\widetilde{C}$ we have
$$
\wprime(\widetilde{C}) \leq \wprime(C_1 \triangle C'_1) 
\leq \wprime(C_1) + \wprime(C'_1) < r/2+r/2 =r.$$
This is a contradiction because~$M$ has no circuit of weight less than~$r$.
\end{proof}
\noindent 
Thus, as we will see, it suffices to bound the number of circuits $C_2$ in $M_2$.
Let $C^*_e$ be the unique choice of a circuit provided by \cref{cla:unique}
(if one exists) for element~$e \in S$.
For the ease of notation,
we assume in the following that there is a~$C^*_e$ for every $e \in S$.
Otherwise we would delete any element $e\in S$ from~$M_2$ for which no~$ C^*_e$ exists,
and then would consider the resulting smaller matroid.
It might actually be that we thereby delete all of~$S$ from~$M_2$.

We define a weight function~$w'$ on~$E_2$ as follows: 
$$
w'(e) := \begin{cases}
         \wprime(C^*_e), & \text{ if } e \in S,  \\
         \wprime(e), & \text{ otherwise}.
	\end{cases}
$$
We now have that any circuit~$C$ of Case~2 can be written as
$C^*_e \triangle C_2$, for some $e \in S$, or $C = C_2$ when $C_1 = \emptyset$.
Because~$C^*_e$ is unique,
the mapping $C \mapsto C_2$ is injective for circuits~$C$ of Case~2.
Moreover,
we have $w(C) = w'(C_2)$.
This follows from the definition in case that $C = C_2$.
In the other case, we have
\begin{equation}
w(C) = w({C^*_e \triangle C_2}) = \wprime(C^*_e) + \wprime(C_2) = w'(C_2).
\label{eq:ww'}
\end{equation}
For the equalities, recall that $w(e) = 0$ for $e\in S$.

We conclude that the number of circuits~$C_2$ in $M_2$ with $w'(C_2) < 3r/2$ 
is an upper bound on the number of Case 2 circuits~$C$ of~$M$ with $w(C) < 3r/2$.
Now, to get an upper bound on the number of circuits in $M_2$, we want to apply induction hypothesis.
We need the following claim.

\begin{claim}
\label{cla:nocircuit}
There is no circuit $C_2$ in $M_2$ with $w'(C_2) < r$.
\end{claim}

\begin{proof}
For the sake of contradiction let $C_2$ be such a circuit.
We show that there exists a circuit~$C'$ in~$M$ 
with $w(C') < r$. 
This would contradict the assumption of the lemma.

Case(i):  $C_2 \cap S = \emptyset$.
Then $C_2 \in \C'_2$ itself yields the contradiction
because it is a circuit of~$M$ and $w(C_2) = w'(C_2) < r$.

Case(ii): $C_2 \cap S = \{e\}$.
By \cref{cla:disjointCircuits},  
the set $C_2 \triangle C^*_e$ is a disjoint union of circuits of~$M$.
Let $C' \subseteq C_2 \triangle C^*_e$ be a circuit of $M$. 
Then, because $\wprime(e)=0$, we have
$$
w(C') \leq w({C^*_e \triangle C_2}) = \wprime(C^*_e) + \wprime(C_2) = w'(C_2) < r.
$$

Case(iii): $C_2 \cap S = \{e_1,e_2\}$. 
By  \cref{cla:disjointCircuits}, 
similar as in case~(ii),
there is a set $C' \subseteq C_2 \triangle C^*_{e_1} \triangle C^*_{e_2}$ that is a circuit of~$M$.
Then, because $\wprime(e_1)=\wprime(e_2)=0$, we have
$$
w(C') \leq w(C_2 \triangle C^*_{e_1} \triangle C^*_{e_2}) \leq  \wprime(C_2) + \wprime(C^*_{e_1})+ \wprime(C^*_{e_2}) = w'(C_2) < r.
$$

Case(iv): $C_2 \cap S = \{e_1,e_2,e_3\}$.
 Since $S$ is a circuit, it must be the case that $C_2 =S$.
Since $C^*_{e_1},C^*_{e_2},C^*_{e_3}$ and $S$ constitute all the circuits of~$M_1$,
the set $ C^*_{e_1} \triangle C^*_{e_2} \triangle C^*_{e_3} \triangle S$ 
contains a circuit~$C'$ of~$M_1$.
Since $\{e_i\} =  C^*_{e_i} \cap S$, for $i=1,2,3$, we know that $S \cap C' = \emptyset$.
Thus, $C' \in \C'_1$ is a circuit of~$M$.
Since $\wprime(e_1) = \wprime(e_2) = \wprime(e_3) = 0$, we obtain that
$$w(C') \leq \wprime(C^*_{e_1}) + \wprime(C^*_{e_2})+ \wprime(C^*_{e_3}) = w'(S) =w'(C_2) < r.$$
This proves the claim.
\end{proof}

\noindent
By \cref{cla:nocircuit}, 
we can apply the induction hypothesis for $M_2$ with the weight function~$w'$.
We get that the number of circuits~$C_2$ in $M_2$ 
with $w'(C_2) < 3 r/2$  is bounded by 
$$
T_1 := 240 \, m_2^{5}.
$$ 
As mentioned above, 
this is an upper bound on the number of circuits~$C$ in~$M$ with $w(C) < 3r/2$ in Case~2.

\subsection*{Case 3: $\wprime(C_1) \geq r/2$}

Since $w(C) = \wprime(C_1) + \wprime(C_2) < 3r/2$,
we have
$\wprime(C_2) < r$ in this case.
We also assume that $C_2 \neq \emptyset$.
Hence, 
there is an $e \in S$ such that $C_1 \in \C_{1,e}$ and $C_2 \in \C_{2,e}$.

Let $T_{2}$ be an upper bound on the number of circuits $C_1 \in \C_{1,e}$ with 
$\wprime(C_1) < 3r/2$, for each $e \in S$.
Let $T_{3}$ be an upper bound on the number of circuits $C_2 \in \C_{2,e}$ with
$\wprime(C_2) < r$, for each $e \in S$.
Because there are $s$ choices for the element $e \in S$,
the number of  circuits $C = C_1 \triangle C_2$ with $w(C) < 3r/2$ in Case~3 
will be at most 
	
\begin{equation}
\label{eq:type2}
s \, T_{2} \, T_{3}.
\end{equation}
To get an upper bound on the number of circuits in $\C_{1,e}$ and $\C_{2,e}$,
consider two matroids $M_{1,e}$ and $M_{2,e}$.
These are obtained from $M_1$ and $M_2$, respectively, by deleting the elements in $S\setminus \{e\}$.
The ground set cardinalities of these two matroids are $m_1-s+1$ and $m_2-s+1$.

We know that for $i=1,2$, 
any circuit~$C_i$ of~$M_{i,e}$ with $e \not\in C_i$ is in~$\C'_{i}$ and 
hence, is a circuit of~$M$.
Therefore, there is no circuit~$C_i$ of~$M_{i,e}$ with $e \not\in C_i$ and $\wprime(C_i) < r$.
Using this fact, we want to bound the number of circuits~$C_i$ of~$M_{i,e}$ with $e \in C_i$.
We start with $M_{1,e}$.

\begin{claim}\label{cla:T1}
An upper bound on the number of circuits~$C_1$ in~$M_{1,e}$ with $e \in C_1$ and
$\wprime(C_1) < 3r/2$ is
\begin{equation}
\label{eq:T1}
T_2  := \min \{8(m_1-s)^{3}, 2^{m_1-s} \}
\end{equation}
\end{claim}

\begin{proof}
Recall that the decomposition of~$M$ was such that~$M_1$ is graphic, cographic or the~$R_{10}$ matroid.

Case(i). When $M_1$ is graphic or cographic, 
the matroid~$M_{1,e}$ falls into the same class by \cref{fac:closed}.
Recall that the ground set of~$M_{1,e}$ has cardinality  $m_1-s+1$. 
In this case, we apply \cref{lem:graphic-set} to~$M_{1,e}$
with $R = \{e\}$ and $\alpha = 3$ and get a bound of
$ 8(m_1-s)^{3}.$
The number of circuits containing~$e$ is also trivially bounded by the number of all subsets
that contain~$e$, 
which is $2^{m_1-s}$.
Thus, we get Equation~(\ref{eq:T1}).

Case(ii). When $M_1$ is the $R_{10}$ matroid, then the cardinality of $M_{1,e}$, that is $m_1-s+1$, is at most~10.
In this case again, we use the trivial upper bound of $2^{m_1-s}$. One can verify that when  $m_1-s+1 \leq 10$ then 
$2^{m_1-s} \leq 8(m_1-s)^{3}$. Thus, we get Equation~(\ref{eq:T1}).
\end{proof}
\noindent 
Next, 
we want to bound the number of circuits~$C_2$ in~$M_{2,e}$ 
with $e \in C_2$ and $\wprime(C_2) < r$. 
This is done in  \cref{lem:circuitsR} below, 
where we get a bound of $T_3 := 48(m_2-s)^2$.

To finish Case~3,
we now have 
\begin{eqnarray*}
T_2  &=& \min \{8(m_1-s)^{3}, 2^{m_1-s} \},\\
T_3 &=& 48(m_2-s)^2.
\end{eqnarray*}
By Equation (\ref{eq:type2}), 
the number of  circuits in Case~3 is bounded by~$s \, T_2 \, T_3$.

\begin{claim}
\label{cla:sT1}
For $s=1,3$ and $m_1 \geq 2s+1$,
$$ s  \, T_2 \, T_3 \leq 2400\,  (m_1-2s)^{3} \, (m_2-s)^{2}.$$
\end{claim}

\begin{proof}
We consider $s \, T_2$.
For $m_1-2s \geq 12$, we have
$$s \cdot 8(m_1-s)^{3} \leq 50 (m_1-2s)^{3}.$$
On the other hand, when $m_1-2s \leq 11 $,
$$s \cdot 2^{m_1-s} \leq 50 (m_1-2s)^{3}.$$
This proves the claim.
\end{proof}

\subsection*{Summing up Cases 1, 2 and 3}

Finally we add the  bounds on the number of circuits of Case~1,~2 and~3.
The total upper bound we get is 
\begin{eqnarray*}
T_0 + T_1 + s \, T_2 \, T_3 
&\leq& 240\, (m_1-2s)^{5}  + 240 \, m_2^{5} +  240\, {5 \choose 2} (m_1-2s)^{3} (m_2-s)^{2} \\
&\leq& 240\, (m_2 + m_1-2s)^{5} \\
&\leq& 240 \, m^{5} 
\end{eqnarray*}
This completes the proof of \cref{thm:regular}, 
except for the bound on~$T_3$ that we show in \cref{lem:circuitsR}.
\end{proof}

\noindent 
Now we move on to prove \cref{lem:circuitsR}, which completes the proof of \cref{thm:regular}.
The lemma is similar to \cref{thm:regular}, but differs in two aspects: 
(i) 
we want to count circuits up to a smaller weight bound, that is, $r$, and 
(ii)
we have a weaker assumption that there is no circuit of weight less than~$r$ 
that does not contain a fixed element~$e$.

\begin{lemma}
Let $M=(E,\I)$ be a connected, regular matroid with ground set size $m \geq 2$ 
and $w\colon E \to \N$ be a weight function on $E$.
Let $r$ be a positive integer and let $\widetilde{e} \in E$ be any fixed element of the ground set. 
Assume that there is no circuit $C$ in $M$ such that $\widetilde{e} \not\in C$ and $w(C)<  r$.
Then,
 the number of circuits $C$ in $M$ such that $\widetilde{e} \in C$ and $w(C) < r$ 
is bounded by $48(m-1)^2$.
\label{lem:circuitsR}
\end{lemma}

\begin{proof}
We closely follow the proof of \cref{thm:regular}.
We proceed again by an induction on~$m$, the size of the ground set~$E$.

For the base case, let $m \leq 10$. 
There are at most $2^{m-1}$ circuits that contain $\widetilde{e}$.
This number is bounded by $48(m-1)^2$, for any $2 \leq m \leq 10$.

For the inductive step,
let~$M = (E,{\mathcal I})$ be a regular matroid with $\abs{E} = m > 10$
and assume that the theorem holds for all smaller regular matroids.
Since $m>10$, matroid~$M$ cannot be~$R_{10}$.
If~$M$ is graphic or cographic, 
then the bound of the lemma follows from \cref{lem:graphic-set}.
Thus, we may assume that~$M$ is neither graphic nor cographic.

By \cref{thm:Seymour}, 
matroid~$M$ can be written as a 1-, 2-, or 3-sum of 
two regular matroids~$M_1 = (E_1,{\mathcal I}_1)$ and~$M_2 = (E_2,{\mathcal I}_2)$.
We use the same notation as \cref{thm:regular},
\begin{eqnarray*}
S &=& E_1 \cap E_2,\\
s &=& \abs{S},\\
m_i &=& \abs{E_i}, \text{ for } i =1,2,\\
\C_i &=& \set{C}{C \text{ is a circuit of } M_i}.
\end{eqnarray*}
The case that~$M$ is a 1-sum of $M_1$ and~$M_2$ is  again trivial.
Hence, we may assume that~$M$ is connected.
By \cref{thm:decomp},
$M$ is a $2$-sum or a $3$-sum of~$M_1$ and~$M_2$,
where~$M_1$ is  a graphic, cographic or the~$R_{10}$ matroid,
and~$M_2$ is a regular matroid containing~$\widetilde{e}$.
For $i=1,2$ and $e \in S$, define
\begin{eqnarray*}
\C_{i,e} &:=& \set{C}{C \in \C_i \text{ and } C \cap S = \{e\}}.
\end{eqnarray*}
Also the weight function~$w$ is extended on~$S$ by $w(e) = 0$, for any $e \in S$.

We again view each circuit $C$ of $M$ as $C_1 \triangle C_2$ and consider cases based on 
how the weight of~$C$ is distributed among~$C_1$ and~$C_2$.
Note that~$\widetilde{e}$ is in~$M_2$ and 
we are only interested in circuits~$C$ that contain~$\widetilde{e}$.
Hence, we have $\widetilde{e} \in C_2$.
Therefore we do not have the case where $C_2 = \emptyset$.
We consider the following two cases.

\begin{description}
\item[{\bf Case (i).}] $\wprime(C_1) < r/2$.
\item[{\bf Case (ii).}]  $\wprime(C_1) \geq r/2$.
\end{description}

\noindent
We will give an upper bound for the number of circuits in each of the two cases.

\subsubsection*{Case (i): $\wprime(C_1) < r/2$} 

Since~$\widetilde{e} \not\in C_1$,
we can literally follow the proof for Case~2 from \cref{thm:regular} for this case.
We have again \cref{cla:unique},
that~$C_1$ is uniquely determined as $C_1 = C_e^*$, for $e \in S$,
or $C_1= \emptyset$.
Therefore the mapping $C \mapsto C_2$ is injective.
The only point to notice now is that the mapping maintains that~$\widetilde{e} \in C$ if and only if $\widetilde{e} \in C_2$.
With the same definition of~$w'$, we also have $w(C) = w'(C_2)$.
Therefore it suffices to get an upper bound on
the number of circuits~$C_2$ in~$M_2$ with $w'(C_2) < r$ and 
$\widetilde{e} \in C_2$.

To apply the induction hypothesis, 
we need the following variant of  \cref{cla:nocircuit}.
It has a similar proof.

\begin{claim}
\label{cla:nocircuit1}
There is no circuit~$C_2$ in~$M_2$ such that $w'(C_2) < r$ and $\widetilde{e} \not\in  C_2$.
\end{claim}

\noindent 
By the induction hypothesis applied to~$M_2$,
the number of circuits~$C_2$ in~$M_2$ with $w'(C_2) < r$ and 
$\widetilde{e} \in C_2$ is bounded by 
$$T_0:=48(m_2-1)^{2}.$$

\subsubsection*{Case (ii): $\wprime(C_1) \geq r/2$}

Since $w(C) = \wprime(C_1) + \wprime(C_2) < r$,
we have $\wprime(C_2) < r/2$ in this case.
This is the major difference to Case~3 from \cref{thm:regular} 
where the weight of~$C_2$ was only bounded by~$r$.
Hence, 
now we have again a uniqueness property similar as in \cref{cla:unique}, 
but for~$C_2$ this time.
A difference comes with~$\widetilde{e}$.
But the proof remains the same.

\begin{claim}
\label{cla:unique2}
For any $e \in S$,
 there is at most one circuit $C_2 \in \C_{2,e}$ with 
$\wprime(C_2) < r/2$ and $\widetilde{e} \in C_2$.
\end{claim}

\noindent 
We conclude that any circuit~$C$ in case~(ii) can be written as
$C = C_1 \triangle C_e^*$, for a $e \in S$ and the unique circuit $C_e^* \in \C_{2,e}$.
Therefore the mapping $C \mapsto C_1$ is injective for the circuits~$C$ of case~(ii).
Thus, it suffices to count circuits $C_1 \in \C_{1,e}$ with $\wprime(C_1) < r$,
for every $e \in S$.

Let $e \in S$ and consider the matroid~$M_{1,e}$
obtained from~$M_1$ by deleting the elements in $S\setminus \{e\}$.
It has $m_1-s+1$ elements.
Since~$M_1$ is a graphic, cographic or~$R_{10}$, 
the matroid~$M_{1,e}$ is graphic or cographic by \cref{fac:closed,fac:R10}.
The circuits in~$\C_{1,e}$ are also circuits of~$M_{1,e}$.

Any circuit~$C_1$ of~$M_{1,e}$ with $e \not\in C_1$ is also a circuit of~$M$.
Thus, there is no circuit $C_1$ of $M_{1,e}$ with $e \not\in C_1$ and $\wprime(C_1) < r$.
Therefore we can apply \cref{lem:graphic-set} to~$M_{1,e}$ with $R = \{e\}$.
We conclude that the number of circuits $C_1 \in \C_{1,e}$ with  $\wprime(C_1) < r$ is at most 
$$T_1:= 4(m_1-s)^{2}.$$
Since there are~$s$ choices for $ e \in S$,
we obtain a bound of~$s \, T_1$.

There is also a trivial bound of $s \, 2^{m_1-s}$ on the number of such circuits. 
We take the minimum of the two bounds. 
Recall from the definition of $2$-sum and $3$-sum that $m_1 \geq 2s+1$.

\begin{claim}
For $s = 1$ or $3$ and $m_1 \geq 2s+1$,
$$\min\{ s \, 2^{m_1-s} , 4s \, (m_1-s)^{2}\} \leq 48(m_1-2s)^2.$$
\end{claim}

\begin{proof}
One can verify that when $m_1 -2s \leq 4 $ then 
$$s \, 2^{m_1-s} \leq 48(m_1-2s)^2.$$
On the other hand, when $m_1 -2s \geq 5 $ then
$$4s \, (m_1-s)^{2} \leq 48(m_1-2s)^2.$$
This proves the claim.
\end{proof}
\noindent 
Hence, we get a bound of $48(m_1-2s)^2$ on the number circuits in case~(ii).
Now we add the number of circuits of case~(i) and~(ii) 
and get a total upper bound of 
\begin{eqnarray*}
48(m_2-1)^{2} + 48(m_1-2s)^2 
		&\leq & 48(m_2-1+m_1-2s)^2 \\
		&\leq & 48(m-1)^2. 
\end{eqnarray*}
This gives us the desired bound and completes the proof of \cref{lem:circuitsR}.
\end{proof}

\bibliographystyle{plain}
\bibliography{tum-postICALP}

\appendix

\section{Proof of \cref{thm:decomp}}
\label{sec:appendix-k-sums}

We show some properties of the sum operation on matroids.
First note that the $k$-sum operations are commutative
because their definition is based on symmetric set differences 
which is commutative.
Further,
it is known that the $k$-sum operations are also associative in some cases. 
We give a proof here for completeness. 
The $2$-sum operation is denoted by~$\twosum$. 
\begin{lemma}[Associativity]
\label{lem:associative}
Let $M = M_1 \twosum M_2$ with $e$ being the common element in $M_1$ and $M_2$.
Let $M_2 = M_3 \triangle M_4$ be a $k$-sum for $k= 2$ or $3$ with the common set $S$. 
Further, let $e \in M_3$.
Then 
\begin{equation}
\label{eq:associative}
M = M_1 \twosum (M_3 \triangle M_4) = (M_1 \twosum M_3) \triangle M_4,
\end{equation}
where $M_1 \twosum M_3$ is defined via the common element $e$ and 
$(M_1 \twosum M_3) \triangle M_4$ is defined via the common set $S$.
\end{lemma}

\begin{proof}
We show that the matroids in Equation~(\ref{eq:associative}) have the same circuits.
This implies the equality.
Let $E_i$ denote the ground set of~$M_i$, for $i=1,2,3,4$.

Let $C$ be a circuit of $M = M_1 \twosum M_2$.
We consider the nontrivial case in \cref{lem:3sum-circuits}:
we have $C = C_1 \triangle C_2$ and $e \in C_1 \cap C_2$, 
where $C_1$ and $C_2$ are circuits in $M_1$ and $M_2 = M_3 \triangle M_4$, respectively.
Similarly,
we have $C_2 = C_3 \triangle C_4$,
for circuits $C_3$ and $C_4$ of $M_3$ and $M_4$, respectively.
By our assumption, we have $e \in C_3$.
It follows that $C_1 \triangle C_3 \subseteq E_1 \triangle E_3$ is a circuit of $M_1 \twosum M_3$.
Since~$C_4$ is a circuit of~$M_4$, 
we get from \cref{cla:disjointCircuits} 
that $(C_1 \triangle C_3) \triangle C_4 $
is a disjoint union of circuits in $(M_1 \twosum M_3) \triangle M_4$.

For the reverse direction, 
consider a circuit $C$ of $(M_1 \twosum M_3) \triangle M_4$. 
Similarly as above by \cref{lem:3sum-circuits},
we can write $C = C' \triangle C_4$,
where $C'$ and $C_4$ are circuits of $M_1 \twosum M_3$ and $M_4$, respectively,
with $S \cap C' = S \cap C_4$.
Further, 
$C' = C_1 \triangle C_3$, where $C_1$ and $C_3$ are circuits in $M_1$ and $M_3$, respectively.
Since $S$ is disjoint from $E_1$, it must be that $S\cap C' = S \cap C_3$.
Thus, $C_3 \triangle C_4 \subseteq E_3 \triangle E_4$ is a union of disjoint circuits in $M_3 \triangle M_4$.
Since, $C_1$ is a circuit in $M_1$, it follows that $C_1 \triangle (C_3 \triangle C_4)$
is a disjoint union of circuits in $M_1 \twosum (M_3 \triangle M_4)$.

Thus, we have shown that a circuit of one matroid in Equation~(\ref{eq:associative})
is a disjoint union of circuits in the other matroid and vice-versa.
Consequently, by the minimality of circuits, it follows that their sets of circuits must be the same. 
\end{proof}

\noindent 
Truemper proves the statement of \cref{thm:decomp} for  $3$-connected matroids.

\begin{definition}[$3$-connected matroid~\cite{Tru98}]
A matroid $M=(E,\I)$ is said to be $3$-connected if for $\ell = 1,2$, and for any partition $E = E_1 \cup E_2 $
with $\abs{E_1},\abs{E_2} \geq \ell$ we have
$$\rank(E_1) + \rank(E_2) \geq \rank(E) +\ell. $$
\end{definition}

\begin{lemma}[Decomposition of a matroid~\cite{Tru98}] 
\label{lem:3connected}
If a binary matroid is not $3$-connected then it can be written as a $2$-sum or $1$-sum of two smaller binary matroids.
\end{lemma}

\begin{theorem}[Truemper's decomposition for $3$-connected matroids, \cite{Tru98}]
\label{thm:decomp3}
Let $M$ be a $3$-connected, regular matroid, that is not graphic or cographic and is not isomorphic to $R_{10}$.
Let $\widetilde{e}$ be a fixed element of the ground set of $M$.
Then $M$ is a $3$-sum  of $M_1$ and $M_2$, where $M_1$  is a graphic or a cographic matroid
and $M_2$ is a regular matroid that contains $\widetilde{e}$.
\end{theorem}

\noindent
\cref{thm:decomp} can be seen as the extension of \cref{thm:decomp3} to connected regular matroids.

\begin{proof}[Proof of \cref{thm:decomp}]
The proof is by induction on the ground set size of $M$.
If $M$ is $3$-connected then the statement is true by \cref{thm:decomp3}. 
If $M$ is not $3$-connected, then we invoke \cref{lem:3connected}.
Since $M$ is connected,
it can be written as $2$-sum of two matroids $M = M_1 \twosum M_2$. 
From the definition of a $2$-sum, it follows that $M_1$ and $M_2$ 
are minors of $M$ (see~\cite[Lemma 2.6]{Sey80}), and thus are regular matroids by \cref{fac:closed}. 
Without loss of generality, let the fixed element $\widetilde{e}$ be in $M_2$.
If $M_1$ is graphic, cographic or $R_{10}$ then we are done. 

Suppose,  $M_1$ is neither of these.  
Let $e'$ be the element common in the ground sets of $M_1$ and $M_2$.
By induction, $M_1$ is a $2$-sum or a $3$-sum $M_1 = M_{11} \triangle M_{12}$,
where $M_{12}$ is a regular matroid that contains $e'$ 
and $M_{11}$ is a graphic or cographic matroid, or a matroid isomorphic to $R_{10}$.
Since $M_{12}$ and $M_2$ share $e'$, we can take the $2$-sum of these two matroids using $e'$.
By \cref{lem:associative}, the matroid $M$ is the same as $M_{11} \triangle (M_{12} \twosum M_{2} )$.
The matroid $M_{12} \twosum M_{2}$ contains $\widetilde{e}$ and is regular because 
both $M_{12}$ and $M_2$ are regular (see \cite[Theorem 11.3.14]{Tru98}).
Thus, the two matroids $M_{11}$ and $M_{12} \twosum M_{2}$ satisfy the desired properties.
\end{proof}

\end{document}